\documentclass{article}
\usepackage{PRIMEarxiv}

\usepackage[utf8]{inputenc} 
\usepackage[T1]{fontenc}    
\usepackage{hyperref}       
\usepackage{url}            
\usepackage{booktabs}       
\usepackage{amsmath, amsthm, amsfonts}       
\usepackage{nicefrac}       
\usepackage{microtype}      
\usepackage{fancyhdr}       
\usepackage{graphicx}       
\graphicspath{{media/}}     
\usepackage{balance} 
\usepackage{algorithm}
\usepackage{algpseudocode}
\usepackage{subcaption}
\usepackage{paralist}
\usepackage{appendix}
\usepackage{float}

\DeclareMathOperator*{\argmax}{arg\,max}
\DeclareMathOperator*{\argmin}{arg\,min}

\newcommand{\ChooseVul}{\textit{CHOOSE-VUL}}
\newcommand{\val}{\textit{val}}
\newcommand{\IS}{IS}
\newcommand{\BS}{BS}

\newtheorem{prop}{Proposition}
\newtheorem{example}{Example}[section]
\newtheorem{theorem}{Theorem}[section]

\title{Learning Effective Strategies for Moving Target Defense with Switching Costs

}

\author{
  Vignesh Viswanathan\thanks{\textit{Equal contribution.}}\\
  \small University of Massachusetts, \small Amherst \\
  \small Amherst, USA \\
  \texttt{\small vviswanathan@umass.edu} \\
   \And
  Megha Bose$^\ast$\\
  \small International Institute of \small Information \\ 
  \small Technology \\
  \small Hyderabad, India \\
  \texttt{\small megha.bose@research.iiit.ac.in} \\
 \And
  Praveen Paruchuri \\
  \small International Institute of \small Information \\ 
  \small Technology \\
  \small Hyderabad, India \\
  \texttt{\small praveen.p@iiit.ac.in} \\
}

\begin{document}
\maketitle

\let\thefootnote\relax\footnotetext{A shorter version of this paper titled \textit{'Moving Target Defense under Uncertainty for Web Applications'} appears as an extended abstract in the \textit{Proceedings of the 21st International Conference on Autonomous Agents and Multiagent Systems (AAMAS 2022), May 9–13,
2022, Online.}}

\begin{abstract}
Moving Target Defense (MTD) has emerged as a key technique in various security applications as it takes away the attacker's ability to perform reconnaissance for exploiting a system's vulnerabilities. However, most of the existing research in the field assumes unrealistic access to information about the attacker's motivations and/or actions when developing MTD strategies. Many of the existing approaches also assume complete knowledge regarding the vulnerabilities of a system and how each of these vulnerabilities can be exploited by an attacker. In this work, we aim to create algorithms that generate effective Moving Target Defense strategies that do not rely on prior knowledge about the attackers. Our work assumes that the only way the defender receives information about its own reward is via interaction with the attacker in a repeated game setting. Depending on the amount of information that can be obtained from the interactions, we devise two different algorithms using multi-armed bandit formulation
to identify efficient strategies. We then evaluate our algorithms using data mined from the National Vulnerability Database to showcase that they match the performance of the state-of-the-art techniques, despite using a lot less amount of information.
\end{abstract}


\section{Introduction}

With the rise in complexity of cyber systems, it is increasingly becoming difficult to know all the vulnerabilities in a particular system beforehand. The ample time attackers typically have to probe for these vulnerabilities has made the situation quite complex for security analysts. The consequences of failure go beyond financial loss and could result in a breach of privacy or a denial of service, significantly harming the users. {\em Moving Target Defense} has established itself as a potential solution to combat threats observed in web applications \cite{Sailik2016Webappmtd, DARE}, operating systems \cite{SMORE}, cloud-based services \cite{Cloud-MTD2, Cloud-MTD}, etc. Instead of maintaining a single configuration, analysts can maintain multiple configurations of the system and alternate between them. The key idea behind this is to take away the advantage the attacker has, i.e., to perform reconnaissance over time on a static configuration. In recent times, the research community has shown a lot of interest in this approach, with a number of works (see \cite{Cho2020Survey}) studying its viability and suggesting methods to switch between the configurations.

Several works make use of the problem's natural model as a game between the analyst (modeled as a defender) and multiple attackers to develop game theoretic approaches for identifying switching strategies, i.e., which configuration to switch to at what phase. However, a lot of the existing research assumes an unrealistic amount of prior knowledge about the vulnerabilities in the system as well as the competency and motivation of the attacker. An effective switching strategy should be unpredictable to the attacker while assuming a realistic amount of prior knowledge. Moreover, it might not be desirable to conduct too frequent switches as the defender incurs a cost for each switch depending on the current and next configuration. 

\subsection{Our Contribution}\label{subsec:contrib}
A key line of research in this topic takes inspiration from and adapts the {\em Strong Stackelberg Equilibrium} solution concept \cite{Sailik2016Webappmtd, Sengupta2020Bsmg}. 
Indeed, the Strong Stackelberg Equilibrium is provably the best strategy to play in a Bayesian Stackelberg Game \cite{Paruchuri2008Dobss}. 
However, computing the Strong Stackelberg Equilibrium requires complete knowledge about the attacker rewards; an unrealistic assumption in most cases. We, therefore, take a different approach and adapt algorithms from the multi-armed bandit literature. This approach helps us design algorithms that do not require knowledge about the attacker rewards while helping us learn good strategies. 

We propose two scalable algorithms that output switching strategies for Moving Target Defense applications - the choice of usage depends on the amount of information available about the system configurations. 
Our first algorithm, \emph{FPL-MTD} (Section \ref{subsec:mtd}), assumes no prior knowledge about the vulnerabilities in the system and the exploits carried out by the attackers.
For the case where we assume such prior knowledge, we propose \emph{FPL-MaxMin} (Section \ref{subsec:maxmin}), which takes advantage of knowledge about the vulnerabilities in the system and knowledge about the attackers. We evaluate our approaches on several datasets to showcase that our approaches match the performance of the state-of-the-art despite using much less information.

\subsection{Related Work}\label{subsec:relwork}
Our work is closely related to and builds upon the work of \cite{Sailik2016Webappmtd}, who introduce the problem of Moving Target Defense for web applications and propose a switching strategy based on the concept of Strong Stackelberg solution.
A follow-up paper for this work also studies a reinforcement learning based approach to generate effective switching strategies \cite{Sengupta2020Bsmg}. Other works that compute an effective switching strategy use genetic algorithms \cite{Crouse2012ImprovingTD, Crouse2011AMT} to improve the diversity of the deployed implementations and reinforcement learning \cite{Zhu2014RobustRL} to minimize the total damage. \cite{Zhu2014RobustRL} also recognize the problem of information uncertainty and assume the same amount of information about the defender reward function as we do in Section \ref{subsec:mtd}. 
To the best of our knowledge, no prior work uses algorithms inspired by the multi-armed bandit literature to generate switching strategies for Moving Target Defense. We empirically compare our approach with all the other utility maximizing approaches described above in Section \ref{sec:expts}.

Other related works for MTD deal with identifying when the defender should deploy a new implementation \cite{Cai2016MovingTD, Zhu2013GameTheoreticAT, MTDPredictability} and how the defender can create multiple implementations to switch between \cite{10.1145/2342441.2342467, AlShaer2012RandomHM, 10.1007/978-3-642-30436-1_32, MigrationMTD, 6900086}. Another line of work studies how Moving Target Defense can be applied in real-world settings \cite{9169999, 9124015}. All these areas of research are orthogonal and complementary to our work.

\vspace{-0.5em}
\section{Preliminaries}\label{sec:prelims}
For any positive integer $t$, we use $[t]$ to denote the set $\{1, 2, \dots, t\}$. 
We present a table consisting of all the notation introduced in this section in Appendix \ref{sec:notation-table}.

\textbf{Game Model: } We model the setting of Moving Target Defense as a {\em repeated bayesian game} played by two different players:
\begin{inparaenum}[(a)]
\item a {\em defender} (denoted by $\theta$) who deploys the system and seeks to protect it, and
\item an {\em attacker} (denoted by $\Psi$) who seeks to exploit the vulnerabilities.
\end{inparaenum}
There usually is only one defender in these games i.e., the organization that deploys the system (which can be a team of analysts). 
However, there can be multiple attackers comprised of individuals or groups of individuals, each with different aims and motivations, trying to exploit the deployed system. As commonly done with Bayesian games \cite{Paruchuri2008Dobss}, we represent these multiple attackers as different attacker types of a single attacker.
More formally, we have a game with $2$ players comprising of one defender and one attacker. The attacker has $\tau$ types; 
we denote these types by a set $\Psi = \{\psi_1, \psi_2, \dots, \psi_{\tau}\}$. We further assume that there exists a probability distribution $ P$ across the set of attackers. This distribution may not always be known to the defender. 

The defender has a set $ C = \{c_1, c_2, \dots c_n\}$ of $n$ system {\em configurations} it can deploy. Various aspects of a system can be changed to construct the different system configurations.
Consider an example of a web application whose database has two different implementations, one using MongoDB and the other using MySQL. If every other part of the web application has only one implementation, such an application will have two different configurations; one with a MongoDB-based database and the other with a MySQL-based database. In addition, assume that the application has been implemented using multiple programming languages -- say one version uses Python and the other uses Java. Depending on compatibility, the defender can now have up to four deployable configurations to switch between -- each corresponding to one possible pair of database implementation and programming language. This example can be extended to a case where every part of the technological stack has multiple implementations, each combination of implementations resulting in a new deployable configuration.
In the case of devices consisting of multiple adjustable architectural layers \cite{device-mtd}, different system configurations can be generated by similar reconfiguration at each layer. 
In a network-based system, as described in \cite{botnet-mtd}, whose aim is to select a k-subset of the nodes to place attack detectors on, different system configurations refer to each potential k-tuple of nodes.

Much like its literal definition, {\em vulnerabilities} refer to the different aspects in the configuration that can be exploited by attackers to violate the integrity of the system. They usually depend on the system's hardware and software features.
Each configuration $c \in  C$ has a set of vulnerabilities $ V_c$.
We define the set of all vulnerabilities by  $ V = \bigcup_{c \in C}  V_c$. 
The vulnerability sets of each configuration need not be disjoint.
They may not always be known beforehand; however, we assume that no configuration is perfect i.e. every configuration has some vulnerability.
We define an {\em exploit} as a method (or an algorithm) that can be used to take advantage of a vulnerability at the cost of the defender. Since every vulnerability has an associated exploit, we use the two terms interchangeably.

\textbf{Rewards: } For each configuration $c \in  C$, for each attacker type $\psi \in \Psi$, for each vulnerability $v \in  V$, we define a reward to the defender $\theta$ at the $t$'th round denoted by $r^t_{\theta}(\psi, v, c)$.
If the attacker type $\psi$ successfully exploits a vulnerability in the deployed configuration, the reward $r^t_{\theta}(\psi, v, c) \in [-1, 0)$; the specific value depends on the amount of damage done by the exploit. 
If the attacker type $\psi$ is unsuccessful in exploiting a vulnerability in the deployed configuration, the reward obtained by the defender is $0$. 
Similarly, we define rewards for each attacker type $\psi \in \Psi$ denoted by $r^{t}_{\psi}(v, c)$. 
$r^{t}_{\psi}(v, c) \in (0, 1]$ if the attacker type $\psi$ successfully exploits a vulnerability in the deployed configuration; $r^{t}_{\psi}(v, c) = 0$ otherwise. Since most of this paper deals with defender rewards, we drop the subscript $\theta$ when referring to defender rewards, denoting them solely by $r^t(.)$.

The different attacker types may result in the defender obtaining different utilities because of the different possible attacker capabilities e.g., some types may not have the expertise to carry out certain attacks or may not be able to inflict as much damage as others. 
It is important to note that the rewards need not be constant throughout; they can be stochastic or even adversarial in nature.

There are $T$ rounds in the repeated game. Each round represents an attempted exploit by an attacker type. It can also represent a fixed time interval after which the defender changes its configuration. We define a mapping $f:[T] \mapsto [\tau]$ that utilizes probability distribution $ P$ over the attacker set to map each round to the attacker type that attacks at that round. 
At the $t$'th round, the defender chooses a configuration to deploy (denoted by $d_t \in  C$) and the attacker type $\psi_{f(t)}$ chooses a vulnerability to exploit (denoted by $a_t^{f(t)}$).
At the end of the round, the defender receives the reward obtained in that round $r_{\theta}^t(\psi_{f(t)}, a_t^{f(t)}, d_t)$ along with some other problem specific information. 

If the defender switches its deployed configuration i.e. $d_t \ne d_{t-1}$, it incurs an additional fixed {\em switching cost}. This cost is given by the mapping $s: C \times  C \mapsto [0,1]$; the cost incurred by switching from configuration $c$ to configuration $c'$ is given by $s(c, c') \in [0,1]$. Since such costs can be determined during the testing of the system, we assume these costs are known a priori. When it comes to web applications, switching costs can model how many technological stack components are different between two configurations, as switching between them would require replacement in those components. For devices with multiple architectural layers \cite{device-mtd}, switching costs can model how expensive (in terms of latency and/or power consumption) changing one architectural layer is as compared to others, resulting in a different cost when switching between configurations that vary in one architectural layer from those that vary in another. This can also be extended easily to model changes in more than one layer. For the network example described earlier \cite{botnet-mtd}, switching cost may depend on how many detectors need to be relocated from one k-tuple to another and/or features of the nodes under consideration.

We make no assumptions about the attacker strategy since the attacker may not be fully rational. We also do not assume knowledge of the attacker utility, since these values will be hard to observe and even harder to find out beforehand. On the other hand, we make the attackers more powerful by assuming that they have access to knowledge about the configurations deployed in the previous rounds and can use it to decide which exploit to attempt. In the bandit literature, this is commonly referred to as an {\em adaptive adversary}.


\textbf{Strategies:} We define a {\em defender strategy profile} by a tuple of length $T$, $(d_1, d_2, \dots, d_T)$ where $d_t$ corresponds to the configuration deployed at timestep $t$. 
We allow $d_t$ to depend on the rewards and the other information observed in all the previous rounds. 

Similarly, we define an {\em attacker strategy profile} by a tuple of length $T$, $(a^{f(1)}_1, a^{f(2)}_2, \dots, a^{f(T)}_T)$ where $a_t^{f(t)}$ corresponds to the exploit attempted at timestep $t$ by attacker type $\psi_{f(t)}$. 
However, unlike the defender, the attacker types may know their rewards a priori. The only additional information each attacker type observes at every round is the configuration used by the defender at the previous rounds. 
It sometimes helps to define $a_t^{f(t)}$ as a function $a_t^{f(t)}(d_1, d_2, \dots, d_{t-1})$ that depends on all the previous defender strategies. 
However, when it is clear from context, we omit the parameters of the function representing the exploit solely by $a_t^{f(t)}$. We also sometimes omit $f(t)$ representing the exploit at round $t$ solely by $a_t$.


\textbf{Objective:} We consider settings where the defender does not have complete a priori information and has to learn from observations at the end of each round of the repeated game. 
In such settings, we aim to create algorithms that generate strategy profiles that minimize the total cost incurred by the defender. 
This cost can be modeled as a simple sum of the rewards and the switching cost at every round of the repeated game. 
More formally, given a strategy $ D = (d_1, d_2, \dots, d_T)$, we define the {\em total utility} of $ D$ as follows
\vspace{-0.5em}
\begin{align*}
    TU( D) = \sum_{t=1}^T (r^t(\Psi_{f(t)}, a_t^{f(t)}, d_t) - s(d_{t-1}, d_t))
\end{align*}
where $r^t$ represents the reward at round $t$ and $s$ represents the switching cost.
Our goal in this paper is to develop algorithms that output strategy profiles that maximize the above expression. 
However, maximizing such a function would not be possible when we do not know the function $a_t^{f(t)}(d_1, \dots, d_{t-1})$ beforehand for each possible $f(t)$. 
Instead, we take an empirical approach and develop algorithms that have a good empirical performance on realistic problem instances.

\section{Multi-armed Bandits and Regret}\label{sec:lower-bounds}

When the rewards at each round are chosen adversarially, it is easy to see that the Moving Target Defense can be modeled as the adversarial multi-armed bandit problem. Each configuration can be seen as an arm of the multi-armed bandit instance, and just like the multi-armed bandit problem, we observe rewards after we deploy a configuration (play an arm). 
This section discusses theoretical results from the multi-armed bandit literature that are applicable to our problem.
For a more detailed summary about bandit algorithms and their theoretical guarantees, we refer our readers to \cite{Slivkins2019Bandits}.

The ability to cast our problem as a multi-armed bandit problem makes it possible to use efficient algorithms like \emph{FPL+GR} \cite{Neuandbartok2013Fpl} or Exp3 \cite{Auer2003Exp3} to obtain strategies with good {\em external regret} guarantees. 
The external regret of a strategy profile can be defined as the difference between the total utility of the strategy profile and the utility of the best fixed strategy in hindsight; a fixed strategy in our context refers to the permanent deployment of a single configuration. 
Both the \emph{FPL+GR} and the Exp3 algorithms are guaranteed to output strategies with an expected external regret of $ O(\sqrt{T})$; the total utilities of such strategies converge to the total utilities of the best fixed strategy in hindsight when $T$ is large.
Indeed, these results extend to games as well. 
For zero sum games, it is well known that average rewards from bandit strategies converge to the value of the game \cite[Chapter 9]{Slivkins2019Bandits}.
\cite{Long2016Fplue} show that a slightly modified \emph{FPL+GR} algorithm minimizes external regret in repeated security games.
These results show that bandit algorithms can generate highly efficient strategies.

However, when switching costs are included in the problem, bandit algorithms lose their theoretical guarantees. The main issue arises from the fact that the algorithms themselves do not consider switching cost; so to use these algorithms for our setting, we must include the switching cost in the reward function. This modification gives us the following regret guarantee:
\begin{align*}
    \max_{d \in  C} & \bigg (\sum_{t=1}^T (r^t(\Psi_{f(t)}, a_t^{f(t)}(d_1, d_2, \dots, d_{t-1}), d) - s(d_{t-1}, d)) \bigg ) \\
    & - \bigg ( \sum_{t=1}^T (r^t(\Psi_{f(t)}, a_t^{f(t)}, d_t) - s(d_{t-1}, d_t)) \bigg ) 
    \le  O(\sqrt{T})
\end{align*}
As one can see from the above expression and as argued in \cite{Dekel2012PolicyRegret}, this notion of external regret does not have any meaning when there are switching costs. For problem instances with switching costs, \cite{Dekel2012PolicyRegret} propose an alternate form of regret called the {\em policy regret}. This regret takes into consideration how the attacker would have behaved if the defender had played a pure strategy and then compares the total utility of the bandit algorithm to the total utility of the best pure strategy in hindsight. The policy regret of a strategy profile is given by the following expression:
\begin{align*}
    \max_{d \in  C} &\bigg (\sum_{t=1}^T r^t(\Psi_{f(t)}, a_t^{f(t)}(d, d, \dots, d), d) \bigg )\\
    &- \bigg ( \sum_{t=1}^T (r^t(\Psi_{f(t)}, a_t^{f(t)}(d_1, d_2, \dots, d_{t-1}), d_t) - s(d_{t-1}, d_t)) \bigg ) 
\end{align*}

\cite{Dekel2012PolicyRegret} show that no algorithm can guarantee a sublinear policy regret against an adversary who can make use of information from prior rounds. In other words, all algorithms will have a policy regret of $\Omega (T)$ in the presence of an adaptive adversary. Since the highest possible value the policy regret can have is $2T$, every algorithm can guarantee a tight regret bound of $ O(T)$. Hence, it is impossible to differentiate algorithms by the strength of their regret guarantee. 

Despite the lack of meaningful theoretical guarantees, due to their performance (both theoretical and empirical) in other settings with repeated games \cite{Long2016Fplue}, we believe popular bandit approaches will still perform well in practice in our problem setting. We, therefore, use these algorithms as inspiration to devise strategies for Moving Target Defense.

\vspace{-0.4em}
\section{Heuristics for Utility Maximization}\label{sec:heuristics}
In this section, we propose two algorithms inspired by algorithms from the multi-armed bandit literature. Both our algorithms do not require information about attacker rewards. We ensure this for two reasons: 

(a) In many cases, the defender has very little prior knowledge about the attacker types and their reward function. An example of such a case would be software that provides novel services that may not have been provided before. Such an application usually has very little data about its vulnerabilities which makes it impossible to infer anything about the different kinds of attackers that may be looking to exploit the application. In addition to this, dynamic applications may not have constant vulnerability sets with new vulnerabilities emerging as the application evolves; the defender can make no assumptions about the vulnerabilities of such an application.

(b) Second, even in the case when the vulnerabilities of each configuration and attacker rewards are known beforehand, the actions of an irrational attacker may not reflect these rewards in any meaningful way. For example, an irrational attacker may not follow the Strong Stackelberg Equilibrium strategy at every round. In such a case, it is better to not assume any strategy on the attacker's part and learn the attacker's strategy through observations.

\vspace{-0.4em}
\subsection{FPL-MTD}\label{subsec:mtd}

Our first algorithm is built for the applications described in point (a) above. 
The algorithm assumes no prior knowledge about the vulnerabilities in the deployable configurations or the different attackers looking to exploit the system.
Algorithm \ref{algo:fpl-mtd} presents steps of \emph{FPL-MTD}.

\emph{FPL-MTD} maintains a reward estimate for each configuration which loosely corresponds to the expected reward that the defender will obtain if they deploy that configuration. 
At each turn, the algorithm either chooses a random arm to explore with some probability or chooses the arm with the highest reward estimate.

\begin{algorithm}
    \caption{FPL-MTD}
    \label{algo:fpl-mtd}
    \begin{algorithmic}[1]
        \State \textbf{Input:} hyperparameters $\eta$ and $\gamma$
        \State $\hat{r}_{c}^1 \gets 0 \quad \forall c \in  C$
        \For{$t$ \textbf{in} $1$ \textbf{to} $T$}
            \State Sample $q \sim \texttt{Bernoulli}(\gamma)$ \label{line:fpl-start}
            \If{$q = 1$}
                \State Let $d_t$ be a uniformly sampled configuration
            \Else 
                \State Sample $z_{c} \sim \texttt{exp}(\eta)$  $\quad \forall c \in  C$
                \State $u_c \gets \hat{r}_{c}^t - z_c \quad \forall c \in C$
                \State $d_t \gets \max_{c \in C}  ({u_c} - s(d_{t-1}, c))$\; \label{line:mtd-selection}
            \EndIf \label{line:fpl-end}
            \State Adversary of unknown type $\psi_{f(t)}$ plays unknown action $a_t$, giving the defender a reward $r^t(\psi_{f(t)}, a_t, d_t)$
            \State Run GR-MTD to obtain $K(d_t)$
            \State $\hat{r}_d^{t+1} \gets \frac{1}{t}{\big [(t-1)\hat{r}_d^t + K(d_t)r^t(\psi_{f(t)}, a_t, d_t) \mathbb{I}\{d = d_t\}\big ]}$\;
        \EndFor
    \end{algorithmic}
\end{algorithm}


\textbf{Configuration Selection: }Inspired by the \emph{FPL+GR} algorithm, our algorithm follows a classic Follow the Perturbed Leader formulation \cite{Hannan1958Fpl, Kalai2005Fpl} to choose configurations. 
The \emph{FPL+GR} algorithm maintains an estimated reward $\hat r_c^t$ for each configuration $c$ and time step $t$. At each timestep $t$, the algorithm perturbs these rewards and chooses the configuration with the maximum perturbed reward.
The key difference between our approach and \emph{FPL+GR} is our consideration of switching costs. 
While \emph{FPL+GR} simply chooses the configuration which maximizes the expected reward, we choose the configuration which maximizes the difference between the estimated average reward and the switching costs that the defender will incur (Line \ref{line:mtd-selection}). 


\textbf{Reward Estimation:} A natural defender reward estimate $\hat r_c^t$ to use for each configuration $c$ is the average reward obtained from deploying $c$ in all the previous time steps. 
More formally, at a time step $t$. we would like to estimate the reward $\hat r_c^t$ as 
\begin{align}
    \hat r_c^t = \frac1{t-1} {\sum_{i \in [t-1]} r^i(\psi_{f(i)}, a_i, c)} \label{eq:mtd-desired-estimate}
\end{align}
where $r^i$ denotes the reward function at time step $i$. Since we switch between configurations, we do not observe the rewards for configurations which are not deployed. Therefore, we cannot use this value exactly. 
However, by randomizing over the set of actions, we can have our reward estimate be an unbiased estimate of the value in \eqref{eq:mtd-desired-estimate}. More specifically, the reward estimate  

\begin{align}
    \hat r_c^t = \frac{1}{t-1} \sum_{i \in [t-1]} \frac{r^i(\psi_{f(i)}, a_i, c)}{\Pr[d_i = c]} \mathbb{I}\{d_i = c\} \label{eq:mtd-unbiased-estimate}
\end{align}

is an unbiased estimate of the desired reward in \eqref{eq:mtd-desired-estimate} where $\Pr[d_i = c]$ represents the probability that $c$ is chosen as the strategy $d_i$ and $\mathbb{I}\{d_i = c\}$ is an indicator function which takes value $1$ if $d_i = c$ and $0$ otherwise.

However, in order to have a good estimate for every configuration, we need to induce randomness in the algorithm which allows different configurations to be chosen at any given time step; in other words, we need to explore. We do this in two different ways. 
First, at every time step, we perturb the reward estimates using an exponential distribution with mean $\eta$. Second, at every time step, we choose a configuration at random with probability $\gamma$. The extent of the induced randomness can be controlled using the hyperparameters $\eta$ and $\gamma$. 

Furthermore, the closed form of the probability distribution $\Pr[d_t]$  cannot be computed efficiently due to the complexity of the algorithm. 
So we replace the value $\frac{1}{\Pr[d_t = c]}$ with an unbiased estimate of it. 
This can be computed using the well-studied method \cite{Neuandbartok2013Fpl} named Geometric Resampling (described in Algorithm \ref{algo:gr-mtd}). 
The key intuition behind this approach is that, by simulating the FPL step of Algorithm \ref{algo:fpl-mtd} again and again till we observe the same configuration, we can simulate a geometric distribution with expected value $\frac{1}{\Pr[d_t = c]}$. Hence, the number of iterations taken by the simulation to terminate is an unbiased estimate of $\frac{1}{\Pr[d_t = c]}$. The detailed description of this subroutine is given in Algorithm \ref{algo:gr-mtd}.

\begin{algorithm}
    \caption{GR-MTD}
    \label{algo:gr-mtd}
    \begin{algorithmic}
        \State $K(d_t) \gets M$
        \For{$k$ \textbf{in} $1$ \textbf{to} $M$}
            \State Follow lines \ref{line:fpl-start} to \ref{line:fpl-end} in Algorithm \ref{algo:fpl-mtd} to produce $\tilde{d}$ as a simulation of $d_t$
            \If{$\tilde{d} = d_t$}
                
                \State $K(d_t) \gets \min(K(d_t), k)$
            \EndIf
        \EndFor
        \State \textbf{return} $K(d_t)$ 
        \end{algorithmic}
\end{algorithm}
In theory, Geometric Resampling can take an infinite time but in practice, we can only run it a finite $M$ times. However, since our algorithm has a uniform exploration component, the probability of choosing any configuration is lower bounded by $\frac{\gamma}{| C|}$.  
Therefore, with $M \ge \frac{| C|T}{\gamma}$, the Geometric Resampling subroutine will terminate with an exponentially high probability of at least $1 - e^{-T}$. 
(The formal statement and proof have been relegated to the appendix.)

\subsection{FPL-MaxMin}\label{subsec:maxmin}
While Algorithm \ref{algo:fpl-mtd} assumes no prior information about the vulnerabilities of the deployable configurations and the attacker types, it is likely that there exist applications that have access to prior information about the vulnerabilities of the system configurations. 
We propose a second algorithm \emph{FPL-MaxMin} for possibly older and well-studied applications where the set of vulnerabilities in the system are known beforehand, along with the number of attackers and their distribution $ P$. The algorithm assumes that the defender observes, along with the rewards, which attacker attacked at a given timestep and which vulnerability the attacker tried to exploit. The only information the algorithm does not observe is the attacker rewards at every round. 

We present steps of \emph{FPL-MaxMin} in Algorithm \ref{algo:FPL-MaxMin}. Similar to \emph{FPL-MTD} we maintain estimates and choose a configuration based on our estimate. However, the estimates we maintain in \emph{FPL-MaxMin} are for each vulnerability - attacker type pair and then we choose configurations using a max-min strategy. This is mainly done to ensure robustness; by playing the max-min strategy, the algorithm guards itself against the worst possible attacker strategy and can, therefore, guarantee a reasonable reward irrespective of the vulnerability exploited by the attacker.

\begin{algorithm}
    \caption{FPL-MaxMin}
    \label{algo:FPL-MaxMin}
    \begin{algorithmic}
        \State \textbf{Input:} hyperparameters $\eta$ and $\gamma$
        \State $\hat{r}_{v,\psi}^1 \gets 0 \quad \forall v \in  V, \psi \in \Psi$
        \For{$t$ \textbf{in} $1$ \textbf{to} $T$}
            \State Sample $q \sim \texttt{Bernoulli}(\gamma)$
            \If{$q = 1$}
                
                \State Let $d_t$ be a uniformly sampled configuration
            \Else 
                \State Sample $z_{v, \psi} \sim \texttt{exp}(\eta)$ for $v \in  V$, $\psi \in \Psi$
                \State $u_c \gets \sum_{\psi \in \Psi} \min_{v \in  V_c}  P_{\psi} (\hat{r}^t_{v, \psi} - z_{v, \psi}) \quad \forall c$
                \State $d_t \gets \max_{c \in C}  (u_c - s_{d_{t-1}, c})$\;
            \EndIf
            \State Adversary of type $\psi_{f(t)}$ plays $a_t$, giving the defender a reward $r^t(\psi_{f(t)},a_t, d_t)$
            \State $p_{v, \psi_{f(t)}} \gets  P_{\psi_{f(t)}} \frac{|i \in [t] : f(i) = \psi_{f(t)} \land a_i = v|}{|i \in [t] : f(i) = \psi_{f(t)}|} \quad \forall v \in  V$
            \State $n_{v} \gets |\{i : v \in  V_{d_i}\}| \quad \forall v \in  V$
            \State $\hat{r}^{t+1}_{v,\psi_{f(t)}} \gets \frac{1}{n_v} \sum_{i \in [t]} \frac{r^i(\psi, v, d_i)}{p_{v. \psi_{f(t)}}} \mathbb{I}\{a_i, \psi_{f(i)} = v, \psi_{f(t)}\} $\;
        \EndFor
    \end{algorithmic}
\end{algorithm}

\textbf{Reward Estimation:} Similar to the previous algorithm, we define a reward estimate for each vulnerability, attacker type pair $(v, \psi)$:

\begin{align}
    \hat r_{v, \psi}^t =& \frac1{|\{i \in [t-1] : v \in  V_{d_i}\}|} {\sum_{i \in [t-1]} r^i(\psi, v, d_i)} \label{eq:maxmin-reward-estimate}
\end{align}

This corresponds to the average reward the defender would have obtained if attacker type $\psi$ exploited vulnerability $v$ at every round when the defender deployed a configuration with vulnerability $v$. Note that $r^i(\psi, v, d_i) = 0$ when $v$ is not a vulnerability of $d_i$.
This approach is compelling since it takes advantage of the fact that multiple configurations might have the same vulnerability, so we can combine the information from attacks on different configurations to estimate the reward of that particular vulnerability.

Since this value cannot be computed accurately, an unbiased estimate for this value can be computed as 
\begin{align}
    \hat r_{v, \psi}^t = &\frac1{|\{i \in [t-1] : v \in  V_{d_i}\}|} \notag \\
    &\quad \sum_{i \in [t-1]} \frac{r^i(\psi, v, d_i)}{\Pr[a_i, \psi_{f(i)}= v, \psi]} \mathbb{I}\{a_i, \psi_{f(i)} = v, \psi\} \label{eq:maxmin-unbiased-estimate}
\end{align}
where $\Pr[a_t, \psi_{f(t)} = v, \psi]$ represents the probability that attacker of type $\psi$ tries to exploit vulnerability $v$ at time step $t$. and $\mathbb{I}\{a_i, \psi_{f(i)} = v, \psi\}$ is an indicator function which takes the value $1$ if attacker type $\psi$ exploits vulnerability $v$ at time step $i$ and $0$ otherwise. 
This is the reward estimate we use. 

Similar to the previous section, $\Pr[a_t, \psi_{f(t)} = v, \psi]$ cannot be computed exactly. However, the reasons for this are different. 
While \eqref{eq:mtd-unbiased-estimate} relies on a probability dependent on the randomization of the algorithm, \eqref{eq:maxmin-unbiased-estimate} relies on the randomization used by each attacker type and the randomization over the set of attacker types. 

Due to this difference, we cannot use Geometric Resampling since we cannot sample exploits from the attackers to infer their mixed strategy. 
Instead, we approximate the mixed strategy of the attacker type by the empirical mixed strategy the attacker type has used so far. 
More formally, we approximate $\Pr[a_t, \psi_{f(t)} = v, \psi]$ as
\begin{align}
    \Pr[a_t, \psi_{f(t)} = v, \psi] \approx  P_{\psi} \frac{|i \in [t] : \psi_{f(i)} = \psi \land a_i = v|}{|i \in [t] : \psi_{f(i)} = \psi|} \label{eq:gr-approx}
\end{align}
With information about $a_t$ and $\psi_{f(t)}$, we also update our estimate of $\Pr[a_i, \psi_{f(i)} = v, \psi]$ for all rounds $i < t$ using the above expression. 
Using these updated probability estimates, we recompute the reward estimates at the end of each round using \eqref{eq:maxmin-unbiased-estimate}.
Note that when an attacker type does not randomize and plays a pure strategy, the reward estimates computed using \eqref{eq:gr-approx} will reflect this pure strategy i.e. the reward estimates will be non-zero only for the pure strategy of the attacker type. The other rewards do not matter since the attacker type will not play any other strategy.

\textbf{Configuration Selection:} When we only have one attacker type (say $\psi$), arguably the most robust approach to choose a configuration using reward estimates for each vulnerability is the max-min strategy --- estimate the value of every configuration as the reward of the least rewarding vulnerability of that configuration. 
Then, choose the configuration with the highest value subject to switching costs. More formally, this evaluates to

\begin{align}
d_t = \argmax_{c \in  C} \min_{v \in  V_c} \big{(} \hat{r}^{t}_{v, \psi} - s(d_{t-1}, c) \big{)}
\end{align}

When there are multiple attacker types and the probability distribution across these types is known, we can make use of this information by estimating the value of a configuration as the weighted average of the least rewarding vulnerabilities of the configuration for each type. We weigh each attacker's estimated reward by the probability of the attacker attacking at the given round. This gives us the following configuration selection strategy:
\begin{align}
     d_t = \argmax_{c \in  C} \bigg ( \bigg [\sum_{\psi \in \Psi}  P_{\psi} \min_{v \in  V_c} \hat r^t_{v, \psi} \bigg ] - s(d_{t-1}, c)\bigg ) \label{eq:max-min}
\end{align}
This is the main configuration selection method we use. We also retain the randomization of the algorithm used in Algorithm \ref{algo:FPL-MaxMin} to create further uncertainty about the defender strategy. 

Lastly, we note that the assumption on prior knowledge of the number of attackers can be relaxed.
In the event the defender does not know about the number of attackers or does not observe which attacker attacked at any given round, we can still use Algorithm \ref{algo:FPL-MaxMin} assuming only one attacker.

\section{Experimental Results}\label{sec:expts}

In this section, we compare our algorithm with existing approaches on datasets constructed using information from the National Vulnerability Database (NVD). The NVD consists of data about vulnerabilities present in the different hardware and software components used in computer systems across the world. This data can be used to construct synthetic problem instances to evaluate our algorithms. This evaluation approach has also been used by other works which study Moving Target Defense for web applications \cite{Sailik2016Webappmtd, Sengupta2020Bsmg}. More specifically, we compare our approach with the following \textbf{defender strategies}:

\textbf{BSS-Q} is a Reinforcement Learning (RL) based approach proposed by \cite{Sengupta2020Bsmg}, that learns a mixed strategy over time but assumes knowledge about the attacker reward obtained at each time step. Since we assume we do not know the attacker reward at each time step, we train this algorithm for $50$ episodes of length $10$ and use the output strategy for evaluation. We train for a small number of episodes since this training process is very slow for large problem instances; each episode involves solving $10$ MIQPs. The training process gives the algorithm complete access to attacker and defender rewards and assumes that the attacker plays the $\epsilon$-greedy strategy, given the mixed strategy of defender.

\textbf{S-OPT} is Stackelberg solution based approach proposed by \cite{Sailik2016Webappmtd}. This algorithm assumes complete knowledge about the attacker and defender rewards and computes a mixed strategy that is used at all the timesteps. 

\textbf{RobustRL} is an RL based approach proposed by \cite{Zhu2014RobustRL}, which is similar to our approach, does not need any prior information about the attackers nor about the vulnerabilities in each configuration.

\textbf{S-Exp3} is a state-of-the-art bandit algorithm that models
switching costs against a non-adaptive adversary \cite{Dekel2012PolicyRegret}, while assuming no prior information about attackers or vulnerabilities. We use the method proposed by \cite{Dekel2012PolicyRegret}, to modify the Exp3 algorithm \cite{Auer2003Exp3} to minimize the policy regret in the presence of switching costs against a non-adaptive adversary.

\textbf{FPL+GR} is a multi-armed bandit algorithm used for importance weighting using only access to samples. It assumes that there is no prior information about the attackers or vulnerabilities and learns from the defender rewards observed \cite{Neuandbartok2013Fpl}.

\textbf{BiasedASLR} is a randomized
algorithm inspired by the Address Space Layout Randomization (ASLR) technique \cite{Pax2003aslr}. While ASLR based techniques typically maximize the entropy of the address selection, which is equivalent to configuration selection in our case
, we additionally bias the random distribution towards configurations whose vulnerabilities have been exploited fewer times. More formally, at round $t$, \emph{BiasedASLR} chooses a random configuration inversely proportional to the number of times a vulnerability of this configuration has been exploited in the first $t-1$ rounds.

\emph{BSS-Q} and \emph{S-OPT} are state-of-the-art algorithms in the case of having full information about the rewards and vulnerabilities a priori.
However, both make unrealistic assumptions about the knowledge of attacker rewards and computing power (both require solving a Mixed Integer Quadratic Program). Moreover, neither are optimal: \emph{BSS-Q} requires an unrealistic amount of training to converge to the optimal solution and \emph{S-OPT}, as presented in \cite{Sengupta2020Bsmg}, is sub-optimal since it models the problem as a single stage normal form Bayesian Stackelberg Game. This modelling does not optimally capture the sequential nature of the game. 

To the best of our knowledge, \emph{RobustRL} is the state-of-the-art technique in the minimal information case. \emph{RobustRL}, \emph{S-Exp3} and \emph{FPL+GR} assume as much access to information as the \emph{FPL-MTD} algorithm. However, none of them take switching costs into account; therefore, at each timestep we give them a modified reward which takes into account the switching cost of the previous iteration: $r^t(\Psi_{f(t)}, a_t, d_t) - s(d_{t-1}, d_t)$. \emph{BiasedASLR} is our only baseline which does not use rewards at all, only counting the number of times attackers have attempted to exploit the vulnerabilities of a specific configuration. We evaluate each of the above algorithms for the following \textbf{attacker strategies}:


\textbf{Best Response:} At each time step, the attacker chooses the best vulnerability to exploit against the empirical mixed strategy of the defender derived from the configurations observed in all the previous timesteps.

\textbf{FPL-UE:} At each time step, the attacker chooses a strategy according to the FPL-UE algorithm \cite{Long2016Fplue}.

\textbf{Stackelberg:} At each time step, the attacker plays the optimal Stackelberg solution computed using the DOBSS algorithm \cite{Paruchuri2008Dobss}.

\textbf{Random:} At each time step, the attacker chooses a vulnerability to exploit uniformly at random from the list of vulnerabilities that the attacker can exploit.


\textbf{Quantal Response (QR):} At each time step, the attacker chooses a vulnerability to exploit using the QR model \cite{McKelvey1995QR}. In the security games literature, the QR model has been used to model human (and bounded-rational) adversaries \cite{Yang2013QR, Nguyen2014Suqr}.

\textbf{Biased Stochastic:} At time step $t$, the attacker chooses a vulnerability to exploit at random with a probability proportional to the number of times the exploit would have been successful in the first $t-1$ rounds. This is the attacker analogue of the BiasedASLR strategy.

We evaluate these defender strategies using the attacker strategies mentioned above on randomly generated datasets whose reward distributions follow the (loosely speaking) reward distribution of the vulnerabilities present in the National Vulnerability Database. By evaluating against a specific attacker strategy (say Best Response), we mean, we evaluate our algorithms in an instance where all the attacker types follow the Best Response strategy. We also study how our algorithms can be used to decide which vulnerabilities to fix in Appendix \ref{sec:identification}. We also evaluate our strategies on randomly generated problem instances and the problem instance used by \cite{Sailik2016Webappmtd} in Appendix \ref{sec:additional-expts}.


We evaluate these defender strategies using the attacker strategies mentioned above on the following datasets:
\begin{inparaenum}[(a)]
\item A toy dataset created by \cite{Sailik2016Webappmtd}, that uses vulnerabilities from the National Vulnerability Database (NVD) to generate rewards (Section \ref{subsec:expts-nvd-small}) and
\item Larger randomly generated datasets whose reward distributions roughly follow the Common Vulnerability Scoring System (CVSS) based reward distribution of the vulnerabilities present in the National Vulnerability Database (Section \ref{subsec:expts:nvd-large}).
\end{inparaenum}

By evaluating using a specific attacker strategy (say Best Response), we evaluate our algorithms in an instance where all the attacker types follow the Best Response strategy. Since defender rewards are negative, for ease of exposition, we plot the difference between the total utility of the algorithm in question and the total utility of the uniform random strategy. We refer to this value as the {\em performance} of an algorithm; it can intuitively be understood as {\em how much better} the algorithm is compared to a strategy that chooses which configuration to deploy uniformly at random. Note that the uniform random strategy is an entropy maximizing strategy and therefore (also) serves as a baseline ASLR technique. 



\vspace{-0.4em}
\subsection{Hyperparameter Selection}\label{subsec:expts-hyperparameter}
In order to select the hyper-parameters $\eta$ and $\gamma$, we generate $5$ random zero sum game instances to test the different hyperparameter combinations; the specific details about dataset generation are omitted due to space constraints. We consider all combinations of $\gamma \in \{0.001, 0.002, 0.003, \dots, 0.02\}$ and $\eta \in \{0.01, 0.02, \dots, 0.1\}$.

For each hyperparameter combination, we run our algorithms on all $5$ zero sum datasets $5$ times for $500$ rounds and compute the average total utility. From this selection method, we identified the ideal combination of parameters for \emph{FPL-MTD} to be $\gamma_{\text{MTD}} = 0.007$ and $\eta_{\text{MTD}} = 0.1$. Similarly, for \emph{FPL-MaxMin} we identified $\gamma_{\text{MaxMin}} = 0.006$ and $\eta_{\text{MaxMin}} = 0.03$. 

\begin{table*}
\centering
\scalebox{0.79}{\begin{tabular}{|c|c c c c c|}
    \hline
    \textit{Algorithm} & \textit{Vulnerabilities} & \textit{Attacker Types} & \textit{Defender Rewards (a priori)} & \textit{Defender Rewards (via observation)} & \textit{Attacker Rewards}  \\
    \hline
    BSS-Q & \checkmark & \checkmark & \checkmark &  & \checkmark \\
    \hline
    S-OPT & \checkmark & \checkmark & \checkmark &  & \checkmark \\
    \hline
    \textbf{FPL-MaxMin} & \checkmark & \checkmark & & \checkmark & \\
    \hline
    \textbf{FPL-MTD} & & & & \checkmark & \\
    \hline
    RobustRL & & & & \checkmark & \\
    \hline
    S-Exp3 & & & & \checkmark & \\
    \hline
    FPL+GR & & & & \checkmark & \\
    \hline
    BiasedASLR & & & & \checkmark & \\
    \hline
\end{tabular}}
\caption{Information requirements of implemented algorithms.}
\label{table:algorithm-info}
\end{table*}

\begin{figure*}[t]
     \centering
     \vspace{-1em}
     \begin{subfigure}[b]{0.30\textwidth}
         \centering
         \includegraphics[width=\textwidth]{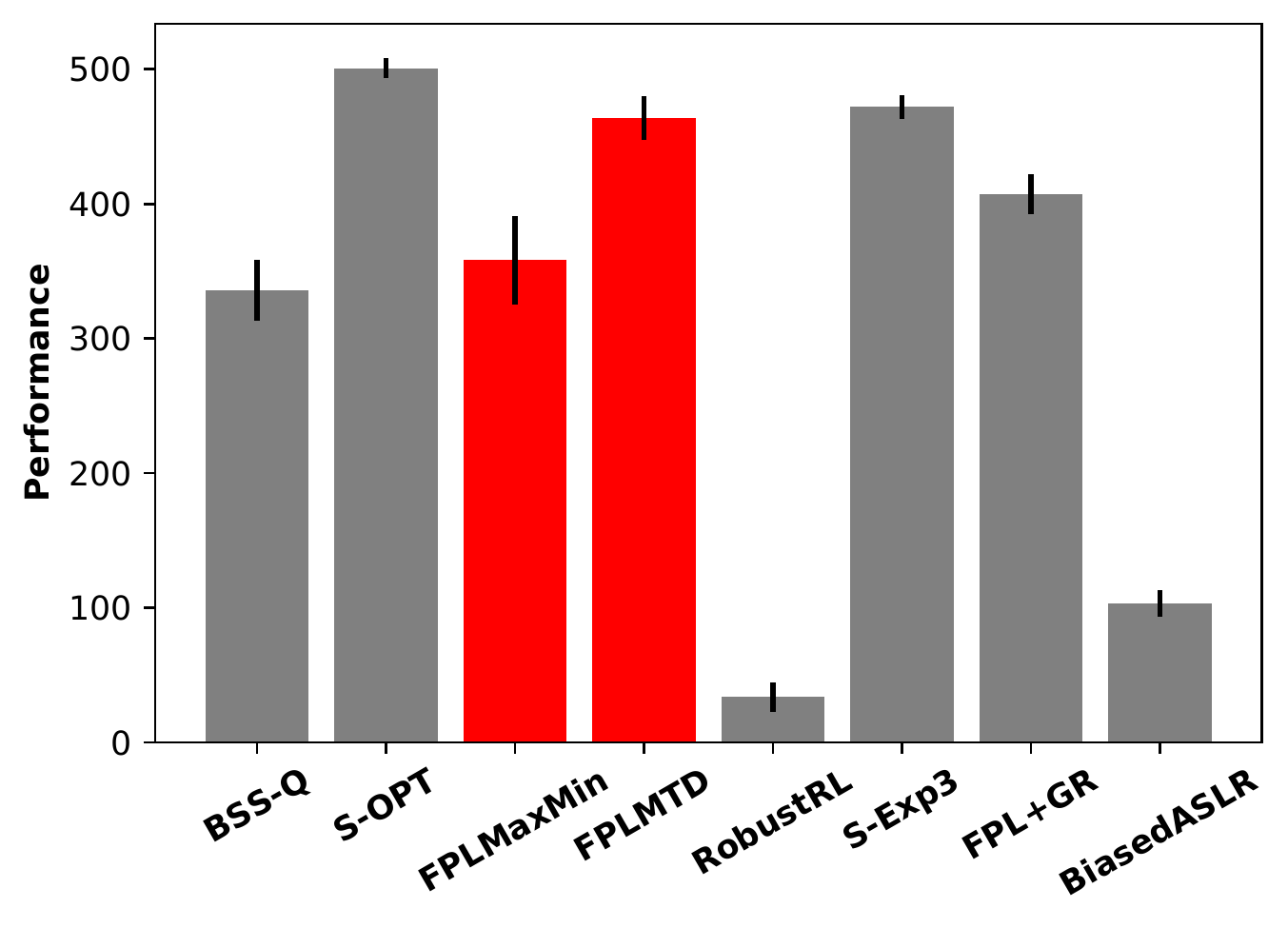}
         \caption{Best Response}
         \label{fig:nvd-best-response}
     \end{subfigure}
     \hfill
     \begin{subfigure}[b]{0.30\textwidth}
         \centering
         \includegraphics[width=\textwidth]{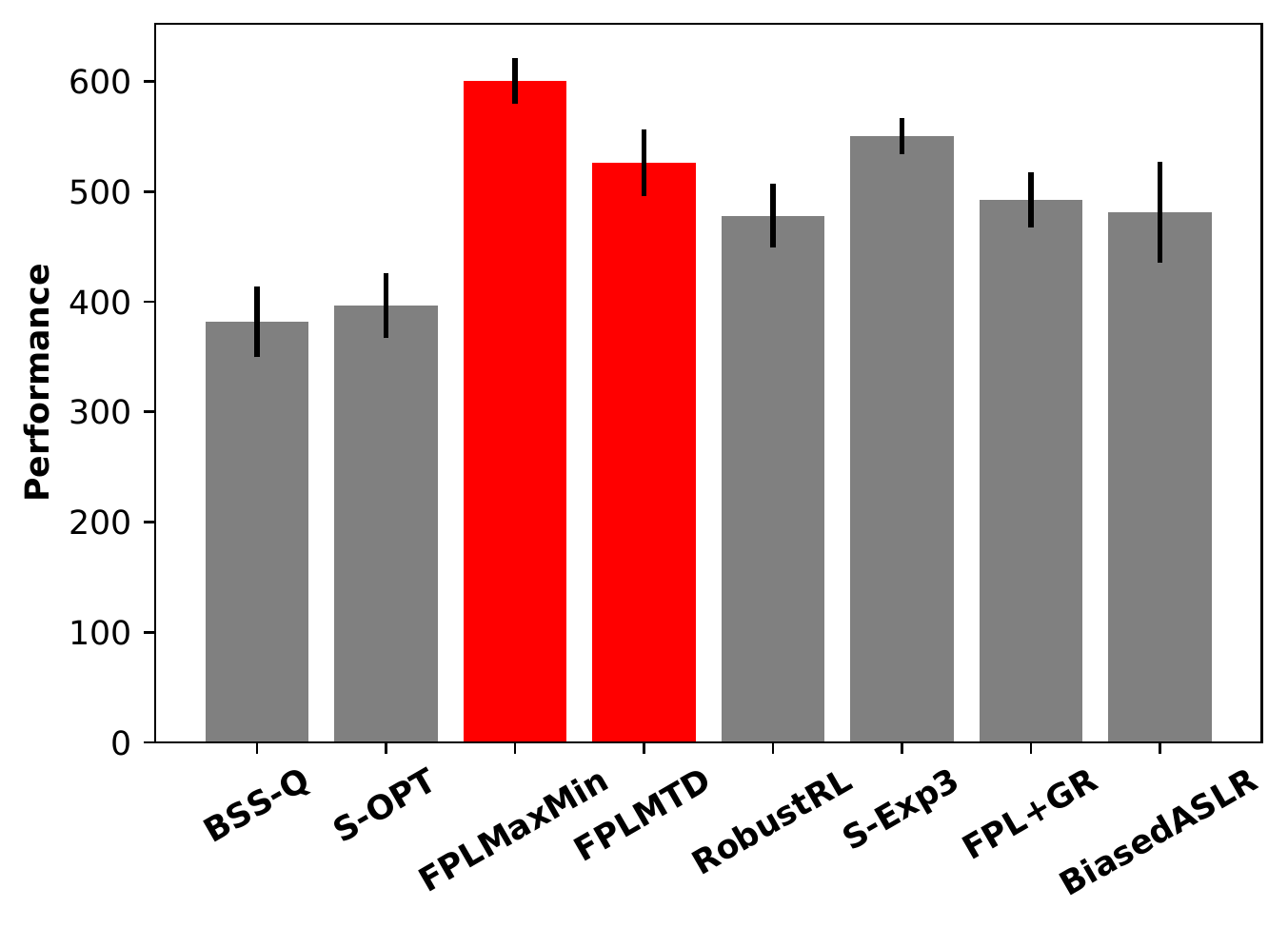}
         \caption{FPL-UE}
         \label{fig:nvd-fpl-ue}
     \end{subfigure}
     \hfill
     \begin{subfigure}[b]{0.30\textwidth}
         \centering
         \includegraphics[width=\textwidth]{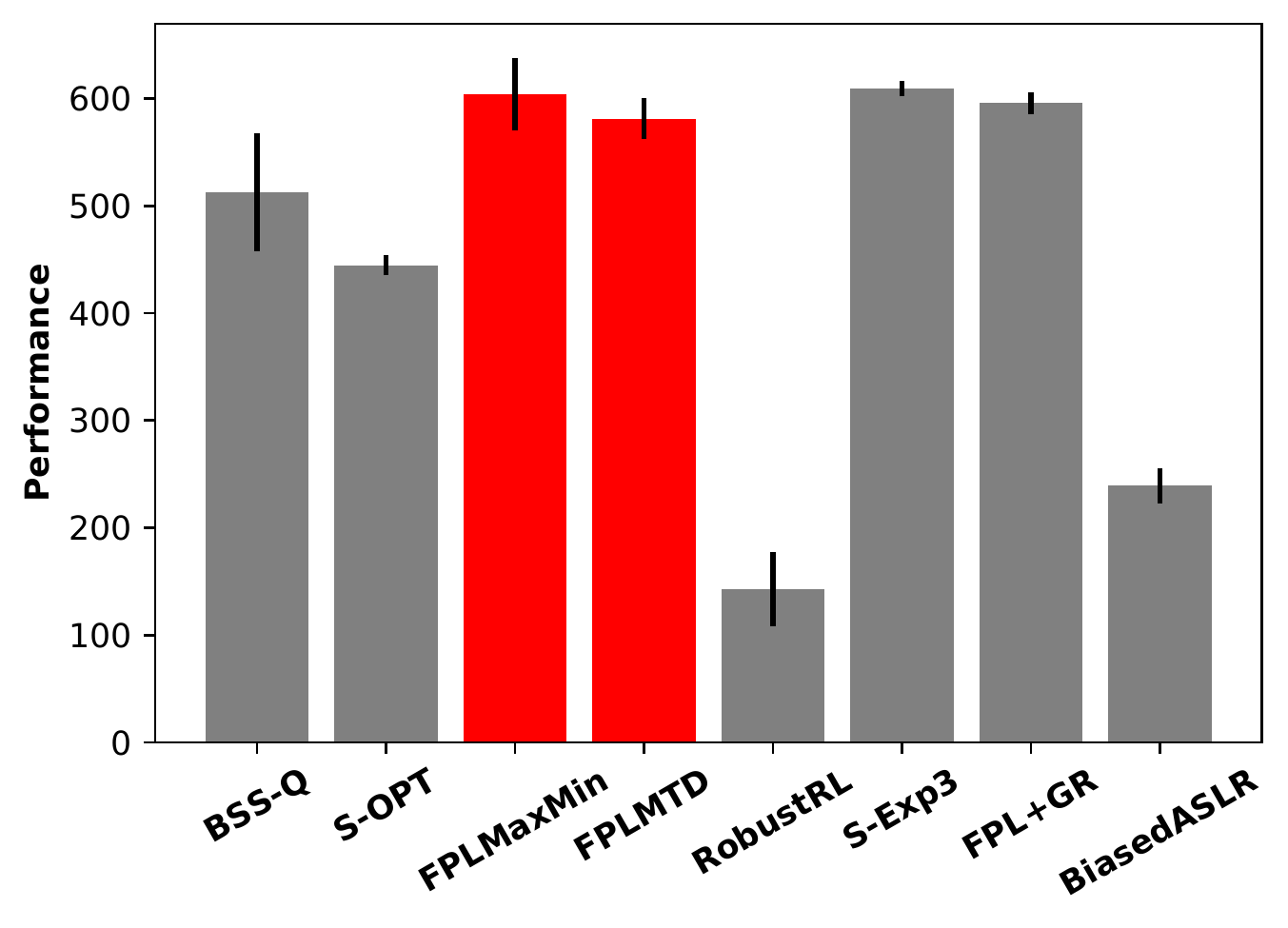}
         \caption{Stackelberg}
         \label{fig:nvd-stackelberg}
     \end{subfigure}
     \hfill
    \newline
     \begin{subfigure}[b]{0.30\textwidth}
         \centering
         \includegraphics[width=\textwidth]{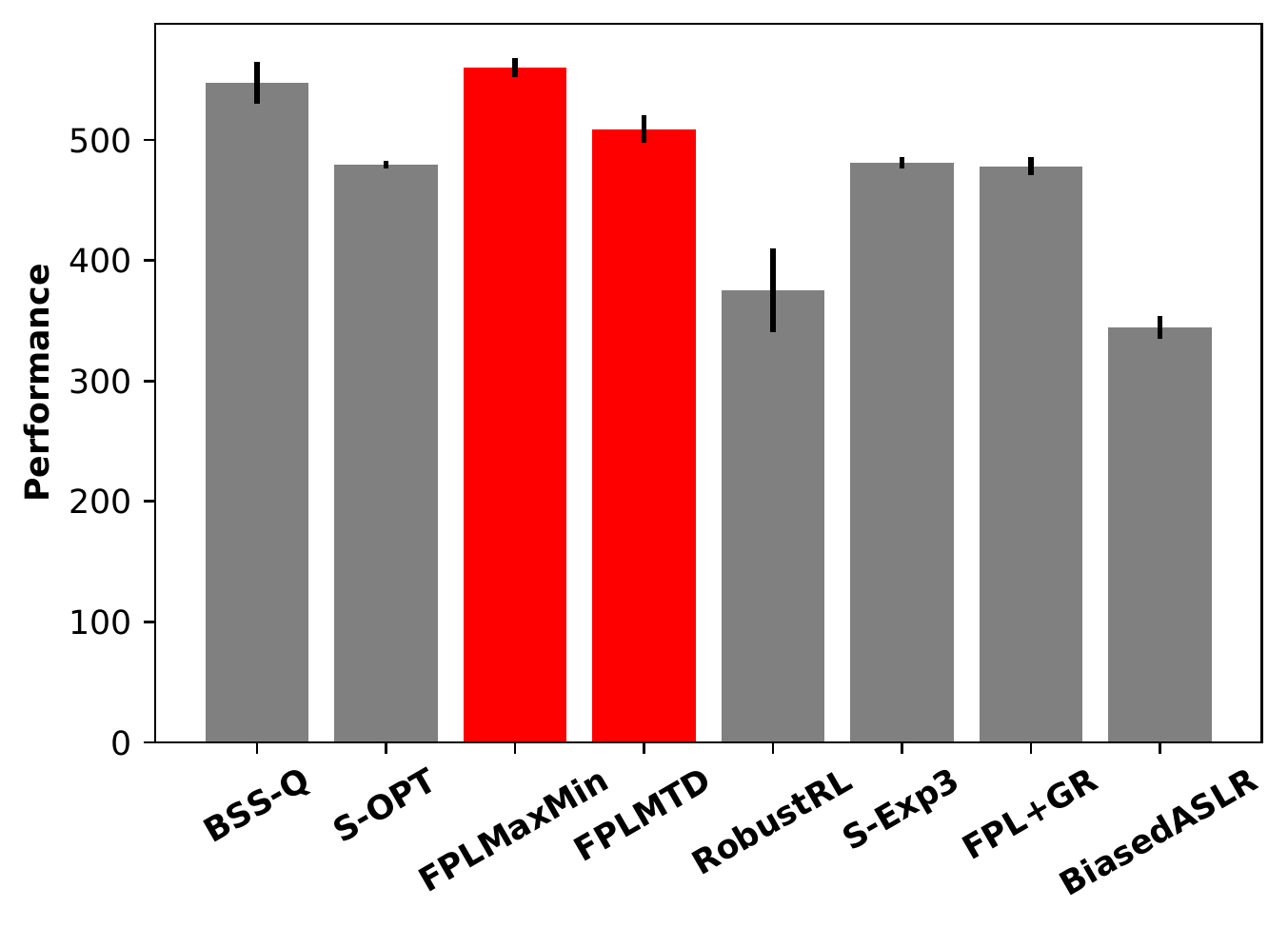}
         \caption{Random}
         \label{fig:nvd-random}
     \end{subfigure}
     \hfill
     \begin{subfigure}[b]{0.30\textwidth}
         \centering
         \includegraphics[width=\textwidth]{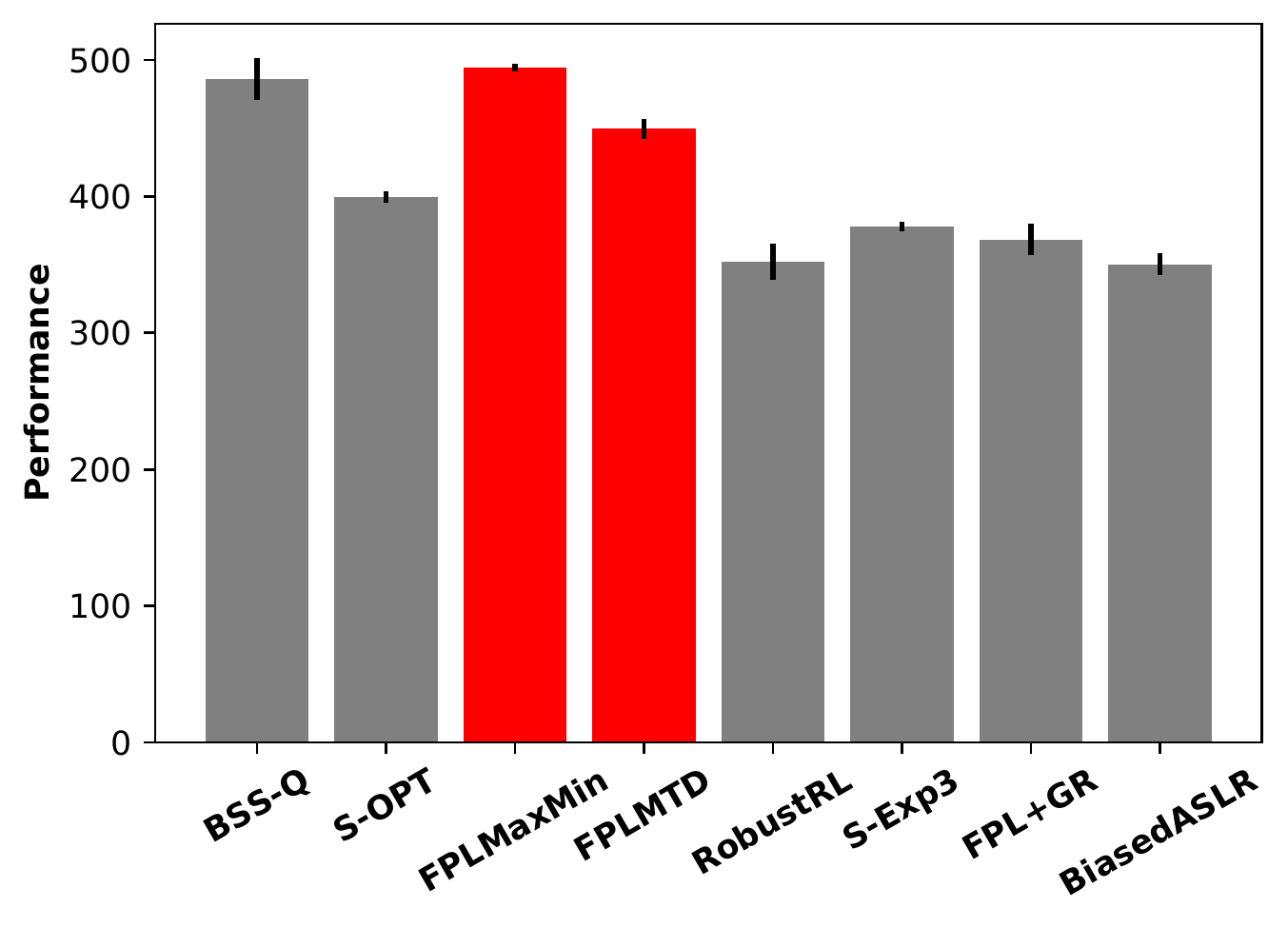}
         \caption{Quantal Response}
         \label{fig:nvd-quantal}
     \end{subfigure}
     \hfill
     \begin{subfigure}[b]{0.30\textwidth}
         \centering
         \includegraphics[width=\textwidth]{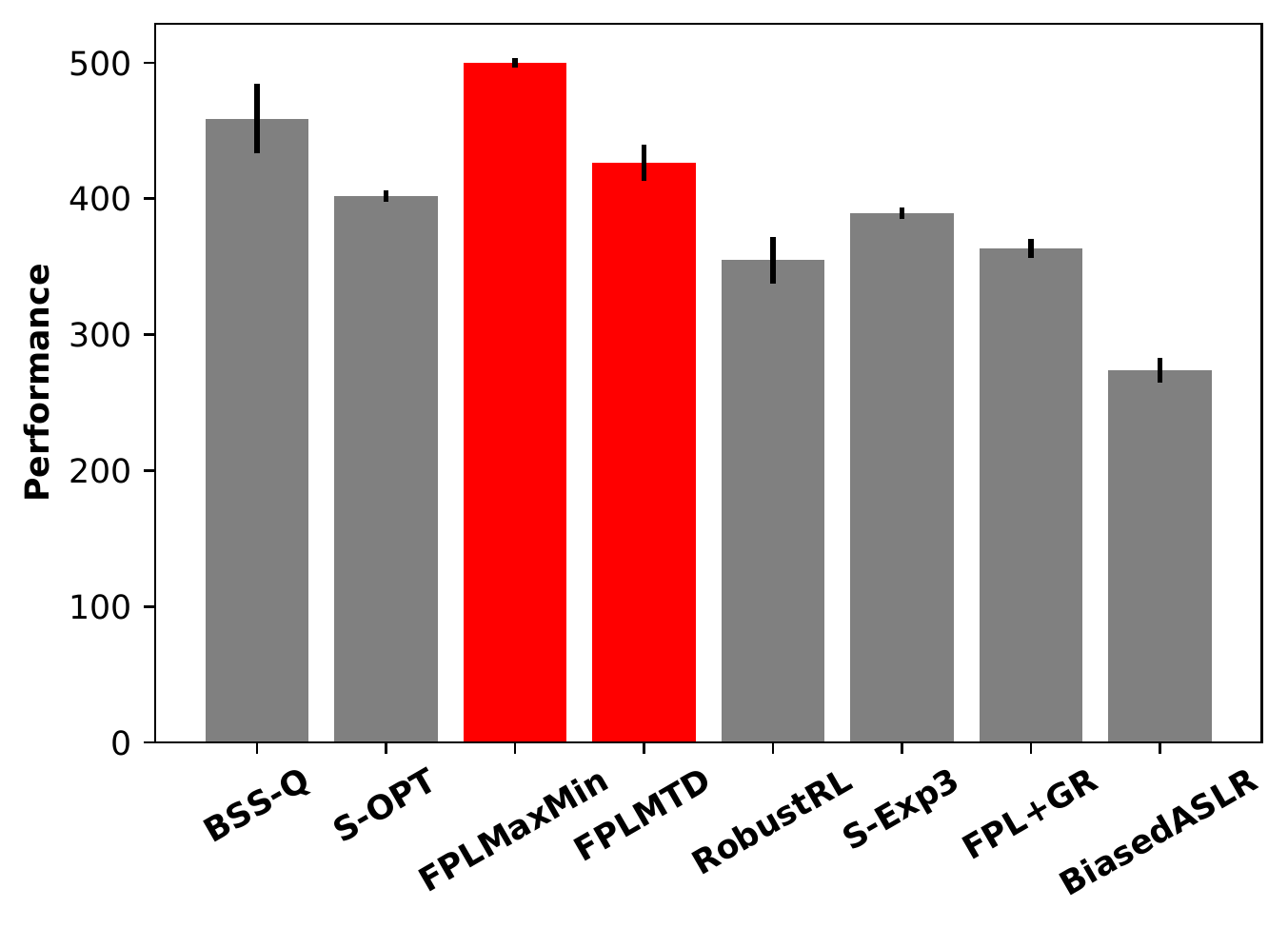}
         \caption{Biased Stochastic}
         \label{fig:nvd-biased}
     \end{subfigure}
        \caption{Performance of algorithms on the small NVD-based data. On each of the graphs, from left to right: BSS-Q, S-OPT, FPLMaxMin, FPLMTD, RobustRL, S-Exp3, FPL+GR, BiasedASLR.
        Performance is defined as the difference between the total utility of the algorithm and the total utility of a uniform random algorithm.
        }
        \label{fig:nvdsmall}
\end{figure*}

\begin{figure*}[t]
     \centering
     \begin{subfigure}[b]{0.30\textwidth}
         \centering
         \includegraphics[width=\textwidth]{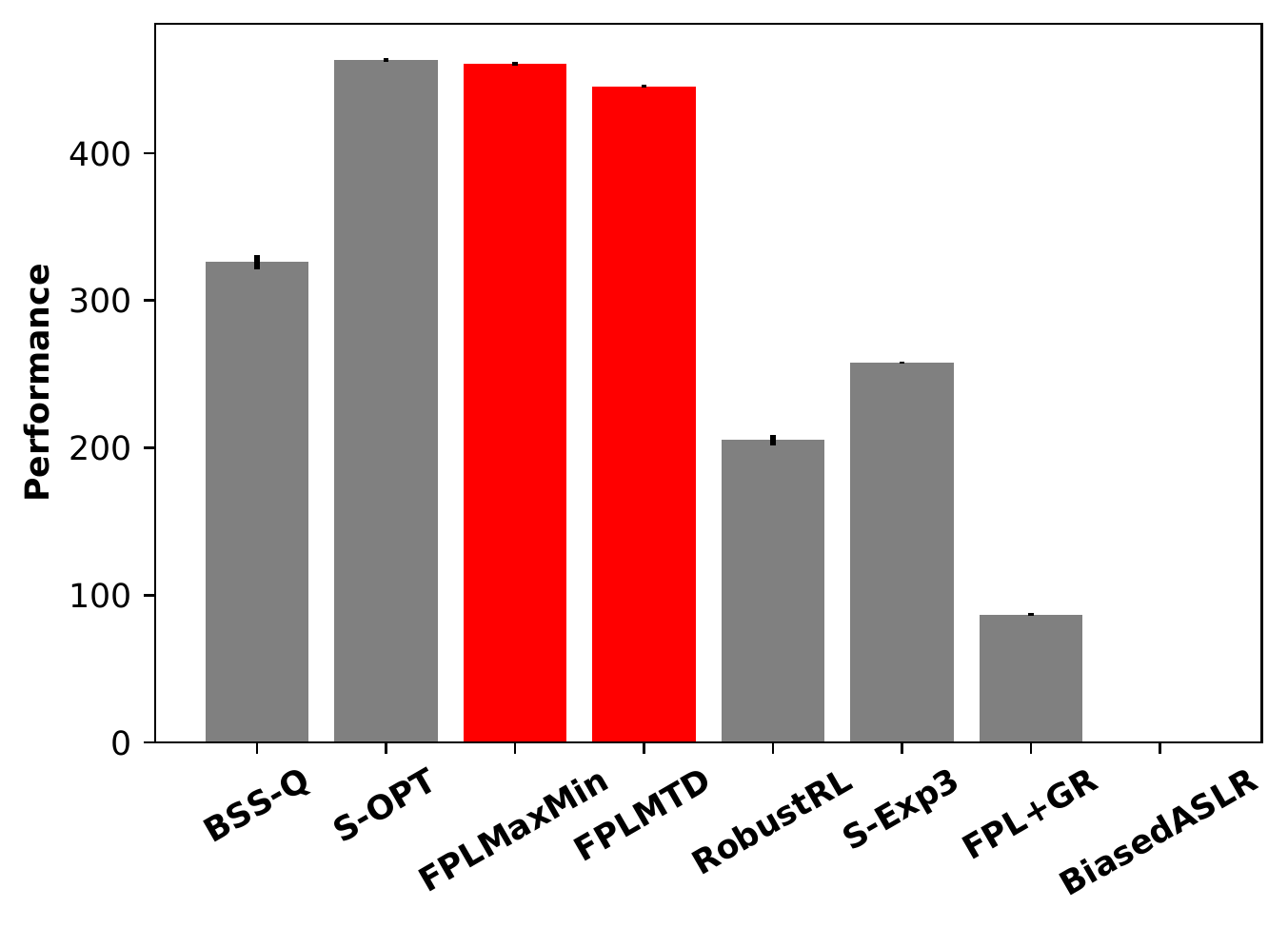}
         \caption{Best Response}
         \label{fig:nvdlarge-best-response}
     \end{subfigure}
     \hfill
     \begin{subfigure}[b]{0.30\textwidth}
         \centering
         \includegraphics[width=\textwidth]{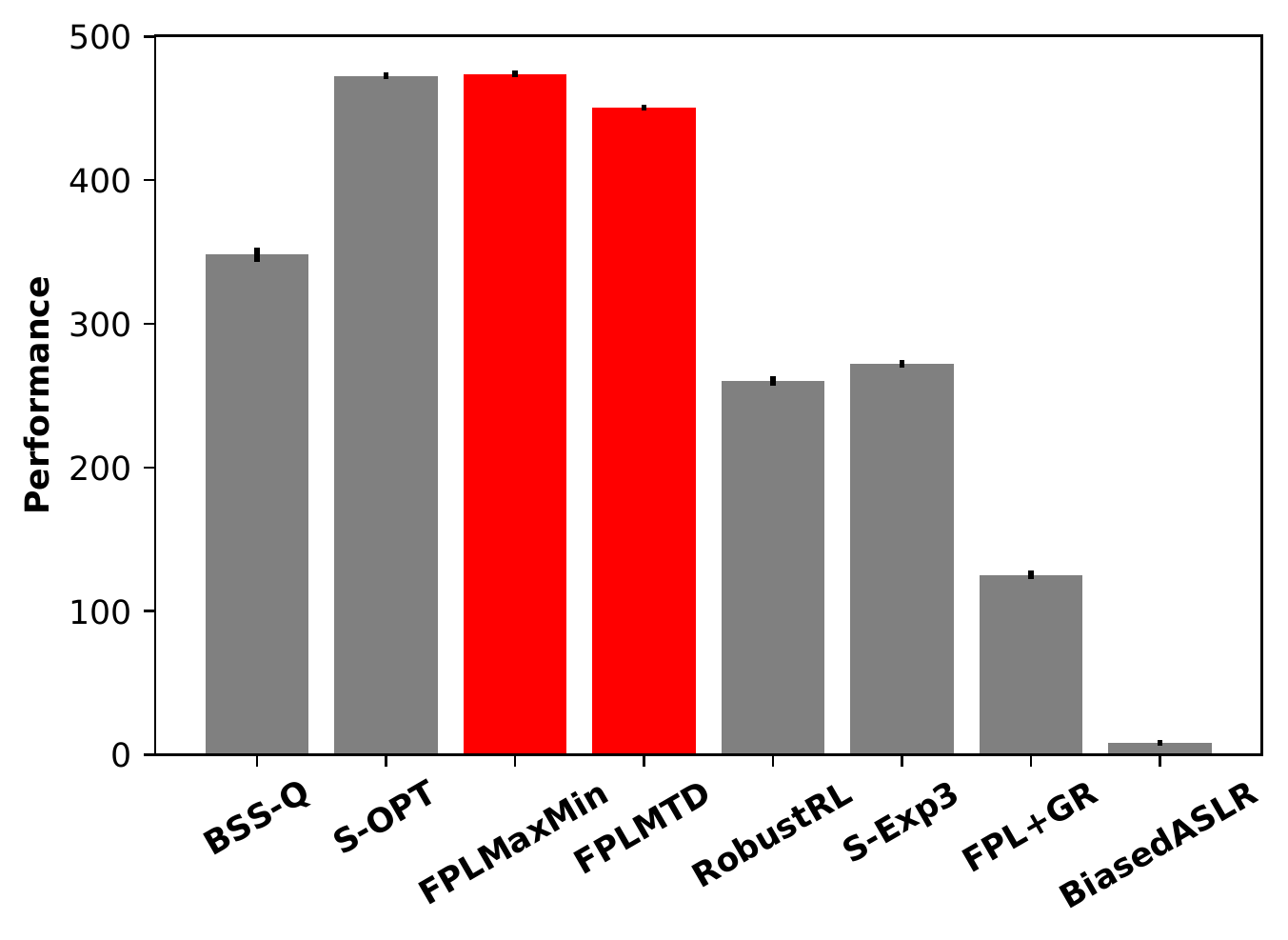}
         \caption{FPL-UE}
         \label{fig:nvdlarge-fpl-ue}
     \end{subfigure}
     \hfill
     \begin{subfigure}[b]{0.30\textwidth}
         \centering
         \includegraphics[width=\textwidth]{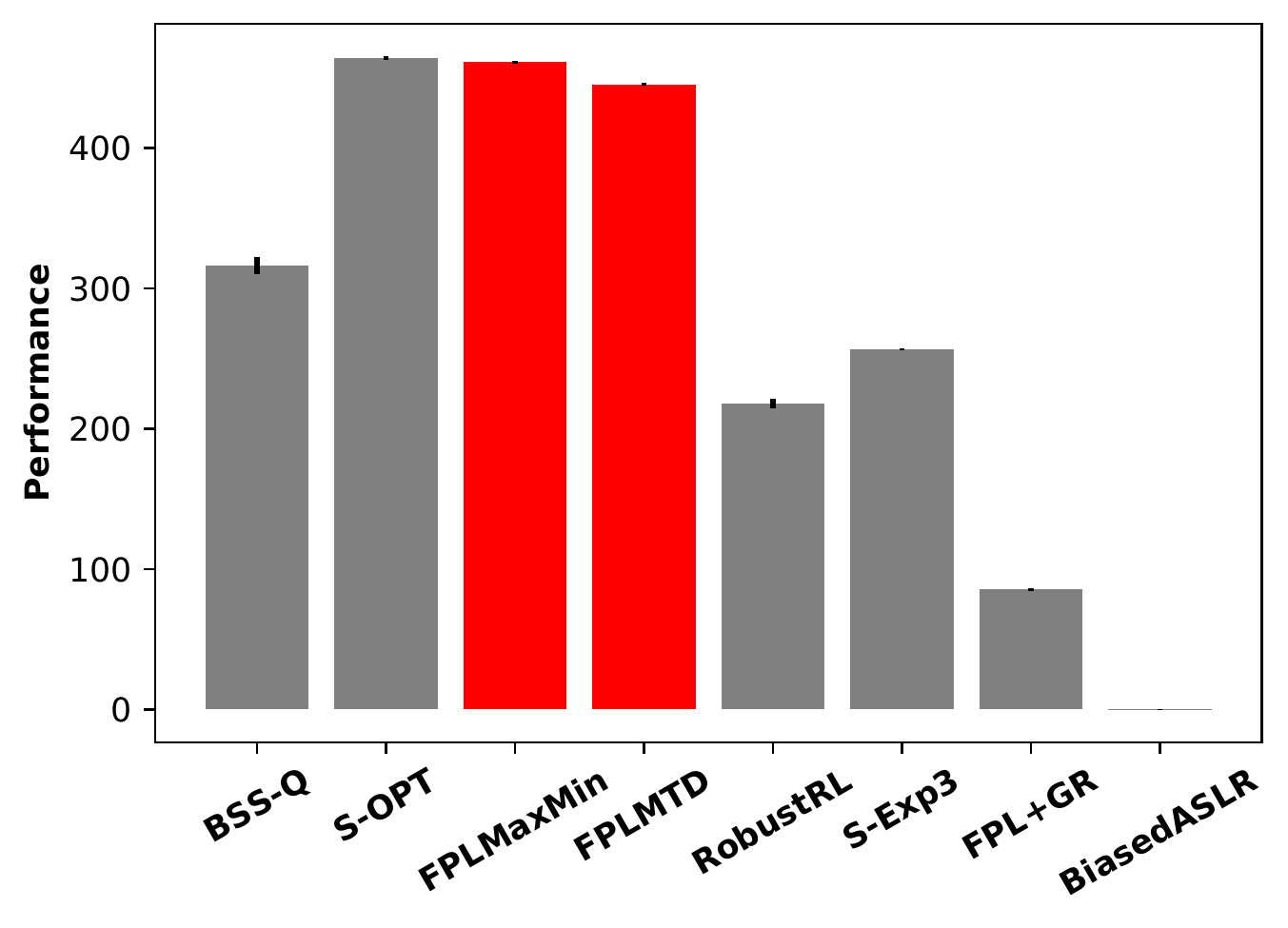}
         \caption{Stackelberg}
         \label{fig:nvdlarge-stackelberg}
     \end{subfigure}
     \hfill
     \newline
     \begin{subfigure}[b]{0.30\textwidth}
         \centering
         \includegraphics[width=\textwidth]{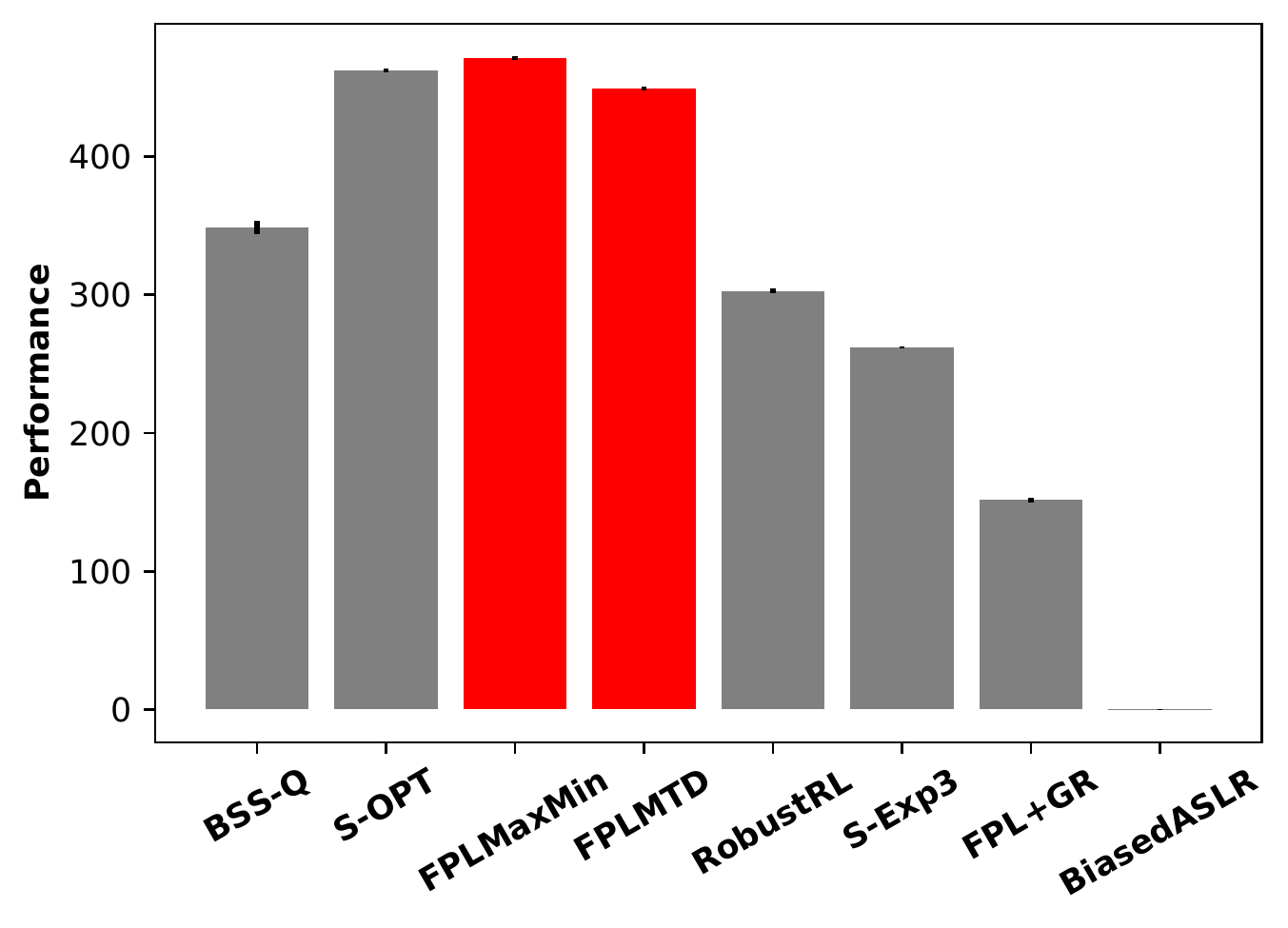}
         \caption{Random}
         \label{fig:nvdlarge-random}
     \end{subfigure}
     \hfill
     \begin{subfigure}[b]{0.30\textwidth}
         \centering
         \includegraphics[width=\textwidth]{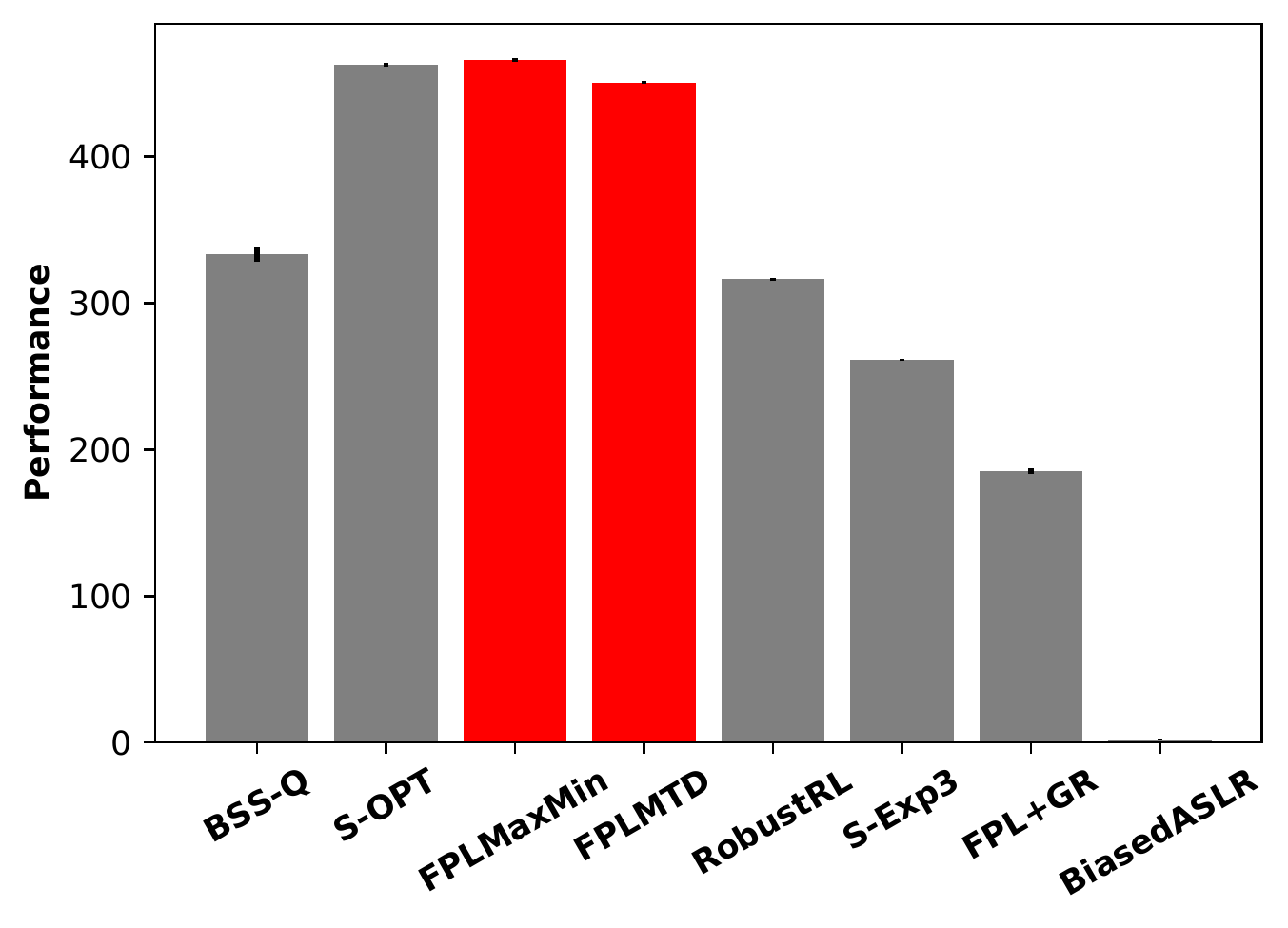}
         \caption{Quantal Response}
         \label{fig:nvdlarge-quantal}
     \end{subfigure}
     \hfill
     \begin{subfigure}[b]{0.30\textwidth}
         \centering
         \includegraphics[width=\textwidth]{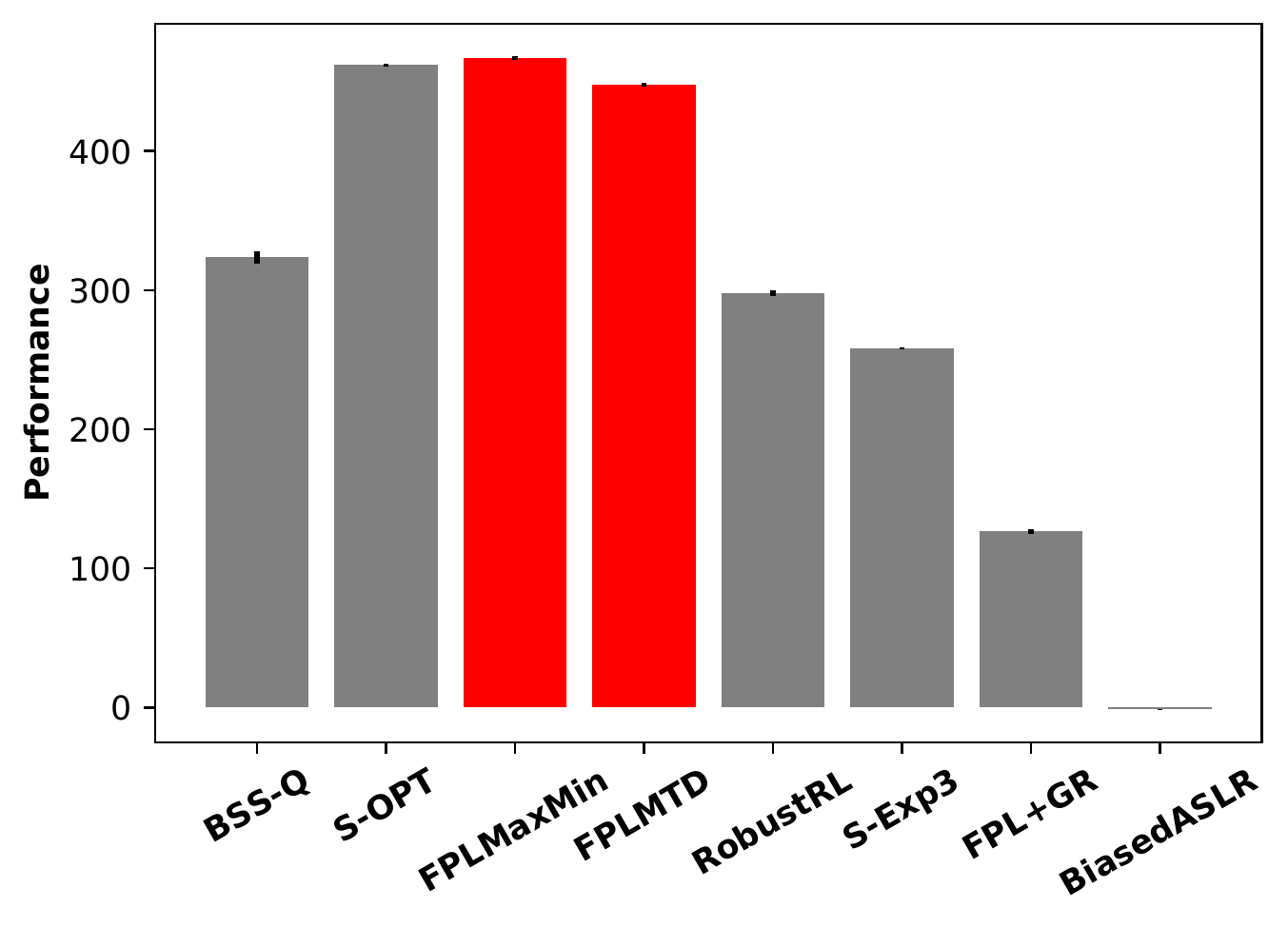}
         \caption{Biased Stochastic}
         \label{fig:nvdlarge-biased}
     \end{subfigure}
        \caption{Performance of algorithms on the large NVD-based data. On each of the graphs, from left to right: BSS-Q, S-OPT, FPLMaxMin, FPLMTD, RobustRL, S-Exp3, FPL+GR, BiasedASLR.
        Performance is defined as the difference between the total utility of the algorithm and the total utility of a uniform random algorithm.
        }
        \label{fig:nvdlarge}
\end{figure*}

\vspace{-0.4em}
\subsection{Small NVD-Based Data}\label{subsec:expts-nvd-small}

This dataset created by \cite{Sailik2016Webappmtd}, contains rewards for $4$ deployable configurations, $3$ attacker types, and slightly less than $300$ vulnerabilities. 
The vulnerabilities and the rewards for each vulnerability are generated using the Common Vulnerability Scoring System of the National Vulnerability Database (for more information, see \cite{Sailik2016Webappmtd}).


Since this dataset is relatively small, we train \emph{BSS-Q} using $100$ episodes. 
We run each defender strategy against each attacker strategy for $1000$ timesteps and plot the average performance over $10$ runs for each algorithm in Figure \ref{fig:nvdsmall}. 
The black lines at the top of each bar correspond to the standard error of the performance.

From the graphs in Figure \ref{fig:nvdsmall}, we infer that no defender strategy conclusively outperforms all the other defender strategies. Among the baselines which do not require prior knowledge, \emph{S-Exp3} outperforms all the other baselines including our algorithm in the case involving the Stackelberg attacker strategy (see Figure \ref{fig:nvd-stackelberg}). For the other attackers, our algorithms either have similar performance or surpass the existing bandit approaches with \emph{FPL-MaxMin} outperforming every other defender strategy against the Random attacker strategy (see Figure \ref{fig:nvd-random}). The state-of-the-art algorithm with full information namely \emph{S-OPT} does not seem to perform well. A deeper inspection suggests that this happens due to a high switching cost. The \emph{S-OPT} approach switches roughly $500$ times in $1000$ timesteps incurring a significant switching cost. This results in the poor performance of \emph{S-OPT} that we see in Figure \ref{fig:nvdsmall}.
 
Based on experimentation using this dataset, we realized that the data is such that two configurations are much more desirable than the other two in terms of the number of vulnerabilities they have. In addition, the instance size is small consisting of only $4$ configurations. Both these factors, make learning and convergence easier for the algorithms. Hence, we present additional results below using a larger dataset.

\vspace{-0.4em}
\subsection{Large NVD-Based Data}\label{subsec:expts:nvd-large}

To design a higher benchmark for the algorithms, we created datasets containing much larger problem instances whose optimal solution would involve switching between several configurations. Larger datasets can allow for a more realistic evaluation due to the presence of several configurations. 
The main limitation when constructing this data was the complexity of the \emph{S-OPT} and \emph{BSS-Q} approaches; computing these strategies takes exponential time and becomes infeasible for much larger game instances.

We generate datasets as follows: we choose the number of configurations uniformly at random between $10$ and $20$, the number of attackers uniformly at random between $3$ and $6$ and the number of vulnerabilities uniformly at random between $500$ and $800$. The reward for each vulnerability is sampled from the empirical distribution of the ``scores'' given to each vulnerability by the National Vulnerability Database. Each dataset corresponds to a synthetic problem instance our algorithms can be evaluated upon. The detailed procedure of dataset generation is presented in Appendix \ref{subsub:nvd-large-data}.

We generate $20$ such datasets and for each dataset, we run each of our algorithms $10$ times for $1000$ timesteps and measure the average performance. We plot the average performance over all the $20$ datasets with standard error bounds in Figure \ref{fig:nvdlarge}. Due to the large number of iterations, note that the standard error is quite low and may not be visible for all the algorithms. This indicates that our estimate of the average performance is close to the true average performance.

The graphs in Figure \ref{fig:nvdlarge} show a clear trend across the different attacker strategies. \emph{S-OPT} performs significantly well, outperforming \emph{BSS-Q}; we attribute this to the slow convergence rate of \emph{BSS-Q} resulting in a lack of convergence to the optimal strategy even after $50$ episodes.  
Across all the four attacker strategies, we find that our algorithm \emph{FPL-MaxMin}, despite using no prior information about the rewards, and algorithm \emph{FPL-MTD} that in addition to utilizing no prior information about rewards, also assumes no prior knowledge about the vulnerabilities of the system, have performance similar to \emph{S-OPT} that uses prior information about all the rewards and vulnerabilities. Our approaches even marginally surpass \emph{S-OPT} in a few cases. 
Our algorithms also perform significantly better than the other baselines \emph{RobustRL}, \emph{S-Exp3}, \emph{FPL+GR} and \emph{BiasedASLR}. More specifically, against the Best Response attacker strategies, \emph{FPL-MTD} has a performance that is $116\%$, $72\%$ and $412\%$ greater than that of Robust-RL, \emph{S-Exp3} and \emph{FPL+GR} respectively. We obtain similar percentages for the other attacker strategies as well (as can be seen in Figure \ref{fig:nvdlarge}). The \emph{BiasedASLR} algorithm has a performance close to $0$ and is significantly outperformed by all other algorithms; this indicates that the algorithm has the same average total utility as that of the uniform random algorithm.

A potential concern while deploying the proposed algorithms in real-world settings can be that they need some number of interactions to learn appropriate reward estimates. During this period, the defender reward and hence the deployed configuration can depend highly on the switching costs. As a result, the deployed configuration may suffer from more attacks while trying to avoid switching. This issue can be resolved by making an implementation choice of setting a higher value for $\gamma$ for a certain number of initial timesteps, thus reducing the effect of the initial resistance to switching.

\section{Conclusions and Future Work}
In this work, we study the problem of generating high quality switching strategies for Moving Target Defense when the defender does not have complete information regarding the vulnerabilities of the system and the strategies of the different attacker types. Building upon the classic multi-armed bandit algorithm FPL, we propose two scalable algorithms: \emph{FPL-MTD} and \emph{FPL-MaxMin}. These algorithms plan for two different levels of information that the defender may have a priori and use them to generate good switching strategies. We showcase, using data from the National Vulnerability Database, that our approaches significantly outperform all the existing algorithms which use the same amount of information as our approaches. They also achieve a slightly better performance than the state-of-the-art, \emph{S-OPT}, despite using much less information than \emph{S-OPT} by eliminating the need to have prior knowledge about vulnerabilities and attacker rewards, and using only the observed defender rewards based on the current attack impact.
In terms of future work, we plan to develop switching strategies for boundedly rational attackers. Algorithms such as SUQR \cite{Nguyen2014Suqr} have been built with this assumption in the context of security games which we plan to adapt here. Another direction of future work is to study bandit-based approaches for additional MTD applications such as intrusion detection systems and network security. 

\bibliographystyle{unsrt}  
\bibliography{references}

\begin{thebibliography}{10}

\bibitem{Sailik2016Webappmtd}
Sailik Sengupta, Satya~Gautam Vadlamudi, Subbarao Kambhampati, Adam Doup\'{e},
  Ziming Zhao, Marthony Taguinod, and Gail-Joon Ahn.
\newblock A game theoretic approach to strategy generation for moving target
  defense in web applications.
\newblock In {\em Proceedings of the 16th International Conference on
  Autonomous Agents and Multi-Agent Systems (AAMAS)}, pages 178–--186, 2017.

\bibitem{DARE}
Michael Thompson, Marilyne Mendolla, Michael Muggler, and Moses Ike.
\newblock Dynamic application rotation environment for moving target defense.
\newblock In {\em 2016 Resilience Week (RWS)}, pages 17--26, 2016.

\bibitem{SMORE}
Joshua~A. Lyle and Nathaniel Evans.
\newblock Software defined networking multiple operating system rotational
  environment.
\newblock 4 2020.

\bibitem{Cloud-MTD2}
Miguel Villarreal-Vasquez, Bharat Bhargava, Pelin Angin, Noor Ahmed, Daniel
  Goodwin, Kory Brin, and Jason Kobes.
\newblock An mtd-based self-adaptive resilience approach for cloud systems.
\newblock In {\em 2017 IEEE 10th International Conference on Cloud Computing
  (CLOUD)}, pages 723--726, 2017.

\bibitem{Cloud-MTD}
Wei Peng, Feng Li, Chin-Tser Huang, and Xukai Zou.
\newblock A moving-target defense strategy for cloud-based services with
  heterogeneous and dynamic attack surfaces.
\newblock In {\em 2014 IEEE International Conference on Communications (ICC)},
  pages 804--809, 2014.

\bibitem{Cho2020Survey}
Jin-Hee Cho, Dilli~P. Sharma, Hooman Alavizadeh, Seunghyun Yoon, Noam
  Ben-Asher, Terrence~J. Moore, Dong~Seong Kim, Hyuk Lim, and Frederica~F.
  Nelson.
\newblock Toward proactive, adaptive defense: A survey on moving target
  defense.
\newblock {\em IEEE Communications Surveys and Tutorials}, 22:709–745, 2020.

\bibitem{Sengupta2020Bsmg}
Sailik Sengupta and Subbarao Kambhampati.
\newblock Multi-agent reinforcement learning in bayesian stackelberg markov
  games for adaptive moving target defense, 2020.

\bibitem{Paruchuri2008Dobss}
Praveen Paruchuri, Jonathan~P. Pearce, Janusz Marecki, Milind Tambe, Fernando
  Ordonez, and Sarit Kraus.
\newblock Playing games for security: An efficient exact algorithm for solving
  bayesian stackelberg games.
\newblock In {\em Proceedings of the 7th International Conference on Autonomous
  Agents and Multi-Agent Systems (AAMAS)}, page 895–902, 2008.

\bibitem{Crouse2012ImprovingTD}
Michael~B. Crouse, E.~Fulp, and D.~Ca{\~n}as.
\newblock Improving the diversity defense of genetic algorithm-based moving
  target approaches.
\newblock In {\em Proceedings of the National Symposium on Moving Target
  Research}, 2012.

\bibitem{Crouse2011AMT}
Michael~B. Crouse and Errin~W. Fulp.
\newblock A moving target environment for computer configurations using genetic
  algorithms.
\newblock In {\em Proceedings of the 4th Symposium on Configuration Analytics
  and Automation (SAFECONFIG)}, pages 1--7, 2011.

\bibitem{Zhu2014RobustRL}
Minghui Zhu, Zhisheng Hu, and Peng Liu.
\newblock Reinforcement learning algorithms for adaptive cyber defense against
  heartbleed.
\newblock In {\em Proceedings of the 1st ACM Workshop on Moving Target Defense
  (MTD)}, page 51–58, 2014.

\bibitem{Cai2016MovingTD}
Guilin Cai, Baosheng Wang, Wei Hu, and Tianzuo Wang.
\newblock Moving target defense: state of the art and characteristics.
\newblock {\em Frontiers of Information Technology \& Electronic Engineering},
  17:1122--1153, 2016.

\bibitem{Zhu2013GameTheoreticAT}
Quanyan Zhu and Tamer Başar.
\newblock Game-theoretic approach to feedback-driven multi-stage moving target
  defense.
\newblock In {\em Proceedings of the 4th Conference on Decision and Game Theory
  for Security (GameSec)}, pages 246--263, 2013.

\bibitem{MTDPredictability}
Richard Colbaugh and Kristin Glass.
\newblock Predictability-oriented defense against adaptive adversaries.
\newblock pages 2721--2727, 10 2012.

\bibitem{10.1145/2342441.2342467}
Jafar~Haadi Jafarian, Ehab Al-Shaer, and Qi~Duan.
\newblock Openflow random host mutation: Transparent moving target defense
  using software defined networking.
\newblock In {\em Proceedings of the 1st Workshop on Hot Topics in Software
  Defined Networks}, page 127–132, 2012.

\bibitem{AlShaer2012RandomHM}
Ehab Al-Shaer, Qi~Duan, and Jafar~Haadi Jafarian.
\newblock Random host mutation for moving target defense.
\newblock In {\em Proceedings of the 8th International Conference on Security
  and Privacy in Communication Networks}, 2012.

\bibitem{10.1007/978-3-642-30436-1_32}
Yulong Zhang, Min Li, Kun Bai, Meng Yu, and Wanyu Zang.
\newblock Incentive compatible moving target defense against vm-colocation
  attacks in clouds.
\newblock In {\em Proceedings of the 27th Information, Security and Privacy
  Conference (SEC)}, pages 388--399, 2012.

\bibitem{MigrationMTD}
Noam Ben-Asher, James Morris-King, Brian Thompson, and William Glodek.
\newblock Attacker skill, defender strategies, and the effectiveness of
  migration-based moving target defense in cyber systems.
\newblock In {\em Proceedings of the 11th International Conference on
  International Warfare and Security (ICCWS)}, pages 21--30, 2016.

\bibitem{6900086}
Michael Thompson, Nathaniel Evans, and Victoria Kisekka.
\newblock Multiple os rotational environment an implemented moving target
  defense.
\newblock In {\em Proceedings of the 7th International Symposium on Resilient
  Control Systems (ISRCS)}, pages 1--6, 2014.

\bibitem{9169999}
Mariusz Rawski, Slawomir Kukliński, Piotr Sapiecha, Marek Pelka, Grzegorz
  Przytuła, Przemysław Wojslaw, and Krzysztof Szczypiorski.
\newblock Mmtd: Mano-based moving target defense for corporate networks.
\newblock In {\em Proceedings of the 5th World Conference on Computing and
  Communication Technologies (WCCCT)}, pages 79--87, 2020.

\bibitem{9124015}
Nico Saputro, Samet Tonyali, Abdullah Aydeger, Kemal Akkaya, Mohammad~A.
  Rahman, and Selcuk Uluagac.
\newblock {\em A Review of Moving Target Defense Mechanisms for Internet of
  Things Applications}, pages 563--614.
\newblock 2020.

\bibitem{device-mtd}
Valentina Casola, Alessandra De~Benedictis, and Massimiliano Albanese.
\newblock A moving target defense approach for protecting resource-constrained
  distributed devices.
\newblock volume 263, pages 22--29, 08 2013.

\bibitem{botnet-mtd}
Sridhar Venkatesan, Massimiliano Albanese, George Cybenko, and Sushil Jajodia.
\newblock A moving target defense approach to disrupting stealthy botnets.
\newblock pages 37--46, 10 2016.

\bibitem{Slivkins2019Bandits}
Aleksandrs Slivkins.
\newblock Introduction to multi-armed bandits.
\newblock {\em Foundations and Trends® in Machine Learning}, 12(1-2):1--286,
  2019.

\bibitem{Neuandbartok2013Fpl}
Gergely Neu and G{\'a}bor Bart{\'o}k.
\newblock Importance weighting without importance weights: An efficient
  algorithm for combinatorial semi-bandits.
\newblock {\em Journal of Machine Learning Research}, 17:1--21, 2016.

\bibitem{Auer2003Exp3}
Peter Auer, Nicol\`{o} Cesa-Bianchi, Yoav Freund, and Robert~E. Schapire.
\newblock The nonstochastic multiarmed bandit problem.
\newblock {\em SIAM Journal of Computing}, 32:48–77, 2003.

\bibitem{Long2016Fplue}
Haifeng Xu, Long Tran-Thanh, and Nicholas~R. Jennings.
\newblock Playing repeated security games with no prior knowledge.
\newblock In {\em Proceedings of the 15th International Conference on
  Autonomous Agents and Multi-Agent Systems (AAMAS)}, page 104–112, 2016.

\bibitem{Dekel2012PolicyRegret}
Ofer Dekel, Ambuj Tewari, and Raman Arora.
\newblock Online bandit learning against an adaptive adversary: from regret to
  policy regret.
\newblock In {\em Proceedings of the 29th International Conference on Machine
  Learning (ICML)}, 2012.

\bibitem{Hannan1958Fpl}
James Hannan.
\newblock Approximation to bayes risk in repeated play.
\newblock {\em Contributions to the Theory of Games}, 3:97--140, 1958.

\bibitem{Kalai2005Fpl}
Adam Kalai and Santosh Vempala.
\newblock Efficient algorithms for online decision problems.
\newblock {\em Journal of Computer and System Sciences}, 71:291--307, 2005.

\bibitem{Pax2003aslr}
Team PaX.
\newblock {PaX} address space layout randomization, 2003.

\bibitem{McKelvey1995QR}
Richard~D. McKelvey and Thomas~R. Palfrey.
\newblock Quantal response equilibria for normal form games.
\newblock {\em Games and Economic Behavior}, 10:6--38, 1995.

\bibitem{Yang2013QR}
Rong Yang, Christopher Kiekintveld, Fernando Ordóñez, Milind Tambe, and
  Richard John.
\newblock Improving resource allocation strategies against human adversaries in
  security games: An extended study.
\newblock {\em Artificial Intelligence}, 195:440--469, 2013.

\bibitem{Nguyen2014Suqr}
Thanh~H. Nguyen, Rong Yang, Amos Azaria, Sarit Kraus, and Milind Tambe.
\newblock Analyzing the effectiveness of adversary modeling in security games.
\newblock In {\em Proceedings of the 27th AAAI Conference on Artificial
  Intelligence (AAAI)}, page 718–724, 2013.

\end{thebibliography}

\appendix
\newpage
\onecolumn

\section{Notation Table}\label{sec:notation-table}
\begin{table}[!h]
    \centering
    \begin{tabular}{|c|l|}
        \hline
        \textbf{Notation} & \textbf{Description} \\
        \hline
        $\theta$ & Defender \\
        \hline
        $\Psi$ & Attacker and the set of attacker types \\
        \hline 
        $\psi_i/\psi$ & Attacker type $i$/an attacker type \\
        \hline 
        $\tau$ & Number of attacker types \\
        \hline 
        $ P$ & Probability distribution across attacker types \\
        \hline
        $ C$ & Set of deployable configurations \\
        \hline
        $c_i / c$ & Configuration $i$/a configuration \\
        \hline
        $ V_c$ & Vulnerability set of configuration $c$ \\
        \hline
        $v_i/v$ & Vulnerability $i$/a vulnerability \\
        \hline 
        $ V$ & Complete set of vulnerabilities \\
        \hline 
        $T$ & Total number of rounds/timesteps \\
        \hline 
        $t$ & A specific time step \\
        \hline 
        $r^t_\theta/r^t$ & Defender reward function at time step $t$ \\
        \hline
        $r^t_{\psi}$ & Attacker reward function for type $\psi$ at time step $t$ \\
        \hline 
        $\psi_{f(t)}$ & Attacker type attacking at time step $t$ \\
        \hline
        $ I_t$ & Additional information set received along with the rewards at the end of every round \\
        \hline 
        $s$ & Switching cost function \\
        \hline 
        $ D / (d_1, d_2, \dots, d_T)$ & Defender strategy profile \\
        \hline
        $d_t$ & Configuration deployed at time step $t$ \\
        \hline
        $(a_1^{f(1)}, a_2^{f(2)}, \dots, a_T^{f(t)})$ & Attacker strategy profile \\
        \hline 
        $a_t^{f(t)}/a_t$ & Exploit attempted at time step $t$ by attacker type $\psi_{f(t)}$ \\
        \hline
        $TU( D)$ & Total utility of a defender strategy $ D$ \\
        \hline 
        $\hat r^t$ & Defender reward estimate at time $t$ \\
        \hline
        $\mathbb{I}(.)$ & Indication function: takes value one if only if expression inside is true \\
        \hline 
        $\gamma, \eta$ & Hyperparameters of FPL-MTD and FPL-MaxMin \\
        \hline 
        $b$ & Budget \\
        \hline
        $p_v$ & Price of fixing vulnerability $v$ \\
        \hline
    \end{tabular}
    \caption{Summary of notations used}
    \label{table:notation}
\end{table}
\section{Missing Results from Section \ref{sec:heuristics}}
\begin{prop}\label{prop:gr}
The Geometric Resampling procedure will terminate with a probability of at least $1 - e^{-T}$ when $M \ge \frac{| C|T}{\gamma}$ where $| C|$ is the number of deployable configurations, $T$ is the total number of time steps and $\gamma$ is the exploration parameter of \emph{FPL-MTD} (Algorithm \ref{algo:fpl-mtd})
\end{prop}
\begin{proof}
At each time step, Algorithm \ref{algo:fpl-mtd} deploys any configuration with a probability of at least $\frac{\gamma}{| C|}$ since it explores with probability $\gamma \in [0, 1]$. 
The Geometric Resampling procedure terminates when the sampled configuration $\tilde{d}$ is same as the actual configuration deployed at that round $d_t$. 
The probability of this event occurring in any given round is lower bounded by $\frac{\gamma}{| C|}$. 
The probability that this event does not occur with $M \ge \frac{| C| T}{\gamma}$ iterations is upper bounded by 
\begin{align*}
    \bigg (1 - \frac{\gamma}{| C|}\bigg )^M \le \bigg (1 - \frac{\gamma}{| C|}\bigg )^{\frac{| C|T}{\gamma}} \le e^{- \frac{\gamma}{| C|}{\frac{| C|T}{\gamma}} } \le  e^{-T}
\end{align*}
The second inequality arises from the fact that $e^{-x} \ge 1 - x$ for any real valued $x$. Replacing $x = \frac{\gamma}{| C|}$ gives us the third expression.
Therefore, the probability of termination is at least $1 - e^{-T}$.
\end{proof}

\section{Additional Experimental Results}\label{sec:additional-expts}
We also evaluate our algorithms on synthetically generated general sum datasets and zero sum datasets. The procedure used for dataset generation is presented in Appendix \ref{subsub:synthetic-data}.

We generate $20$ such general sum and zero sum datasets. For each dataset, we run each of our algorithms $10$ times for $1000$ timesteps and measure the average performance. We plot the average performance over all the $20$ datasets with standard error bounds in Figures \ref{fig:general-sum} and \ref{fig:zero-sum}.

The trends are the same as that of the large NVD-based dataset. For a discussion on these trends, please refer to Section \ref{subsec:expts:nvd-large}.

\begin{figure*}[!tp]
     \centering
     \begin{subfigure}[b]{0.24\textwidth}
         \centering
         \includegraphics[width=\textwidth]{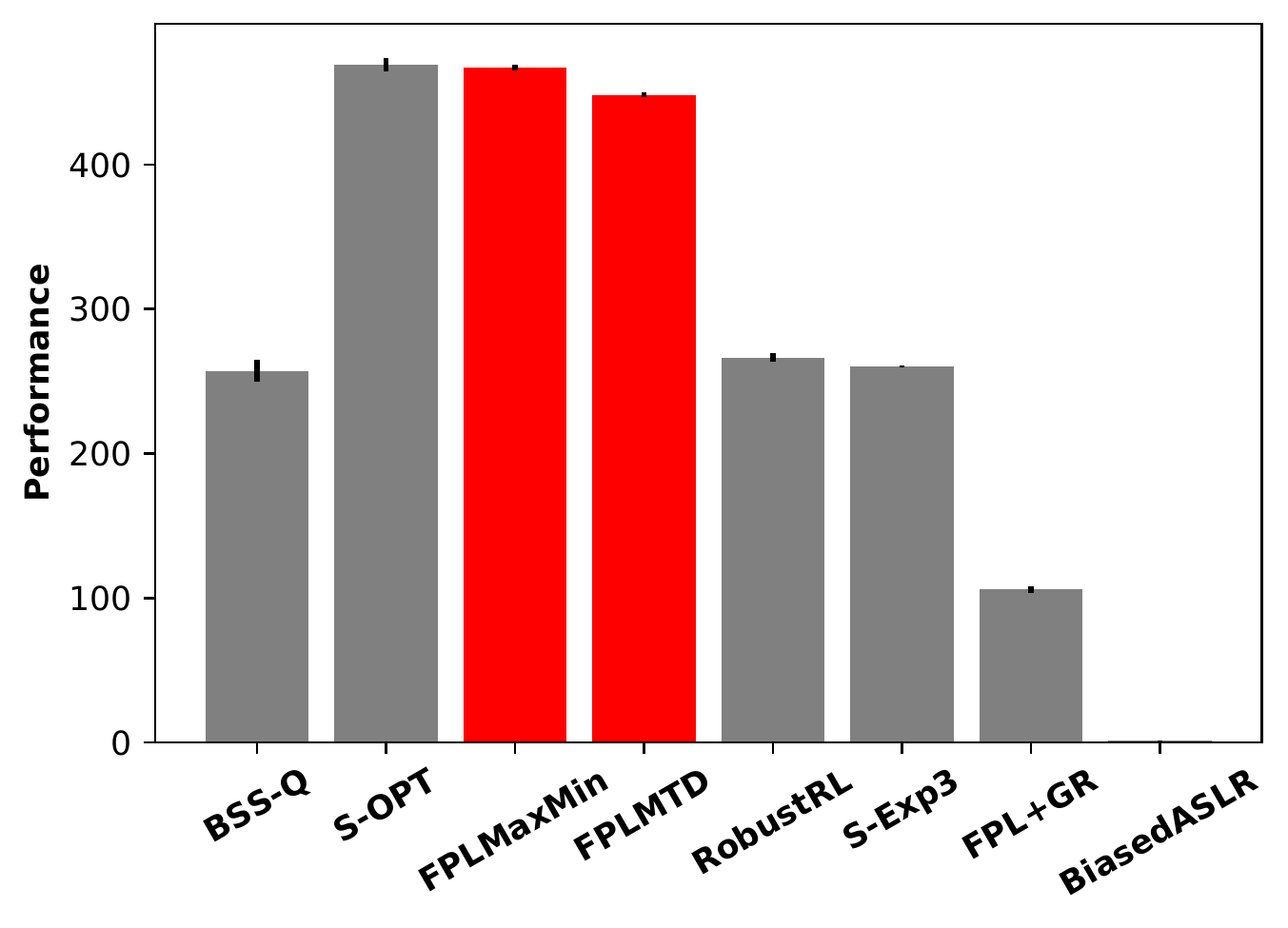}
         \caption{Best Response}
         \label{fig:general-sum-best-response}
     \end{subfigure}
     \hfill
     \begin{subfigure}[b]{0.24\textwidth}
         \centering
         \includegraphics[width=\textwidth]{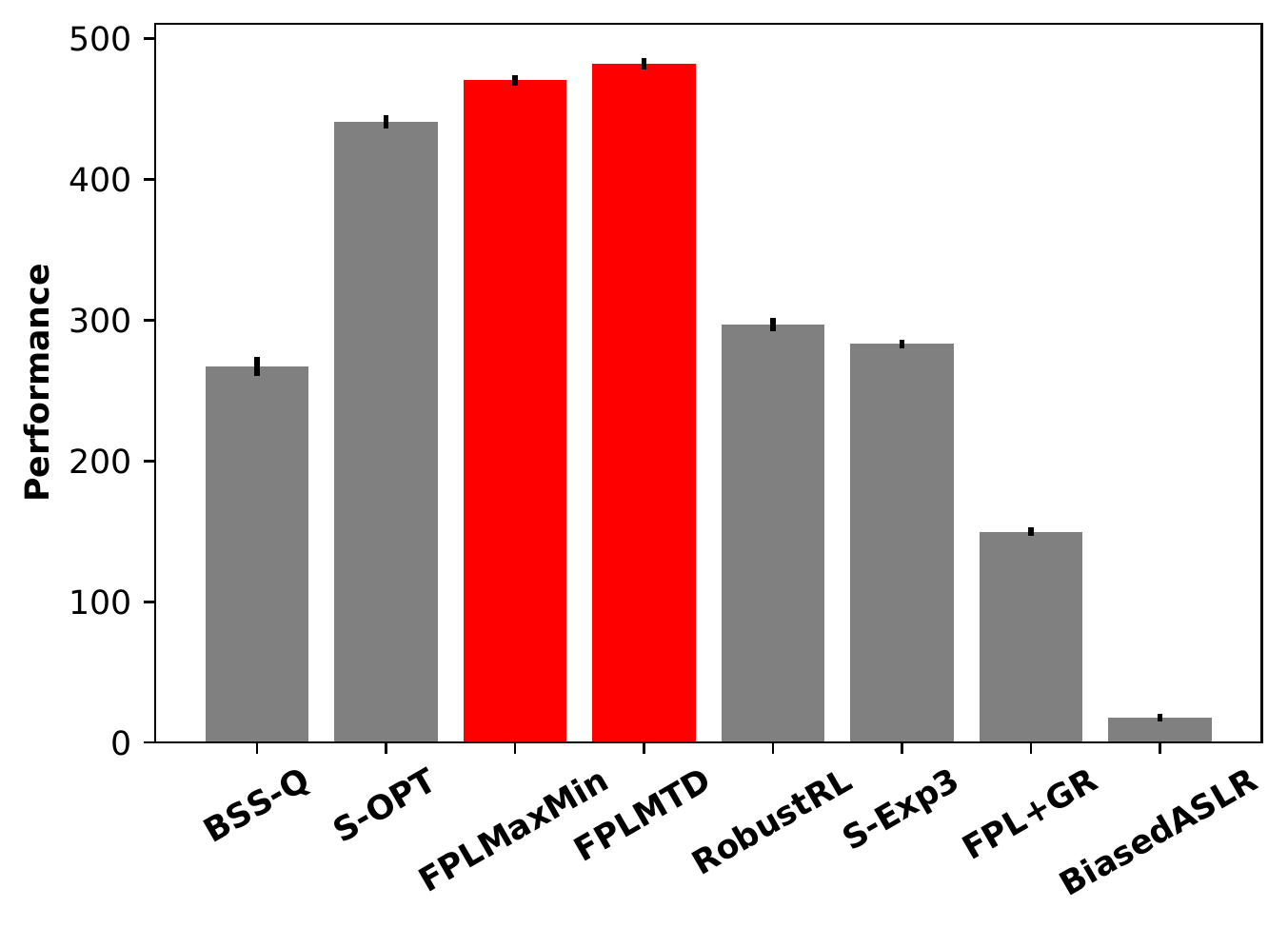}
         \caption{FPL-UE}
         \label{fig:general-sum-fpl-ue}
     \end{subfigure}
     \hfill
     \begin{subfigure}[b]{0.24\textwidth}
         \centering
         \includegraphics[width=\textwidth]{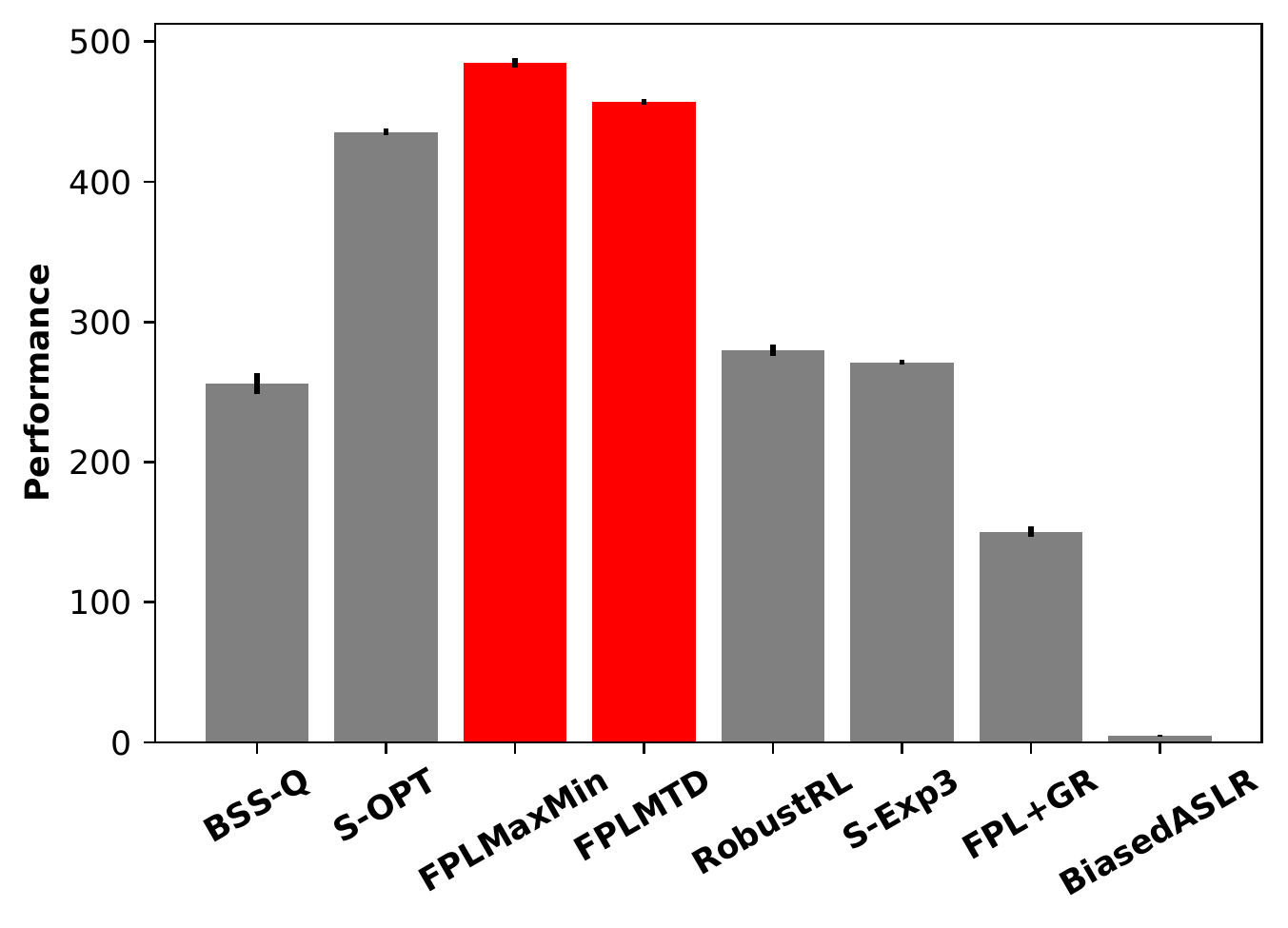}
         \caption{Stackelberg}
         \label{fig:general-sum-stackelberg}
     \end{subfigure}
     \hfill
     \begin{subfigure}[b]{0.24\textwidth}
         \centering
         \includegraphics[width=\textwidth]{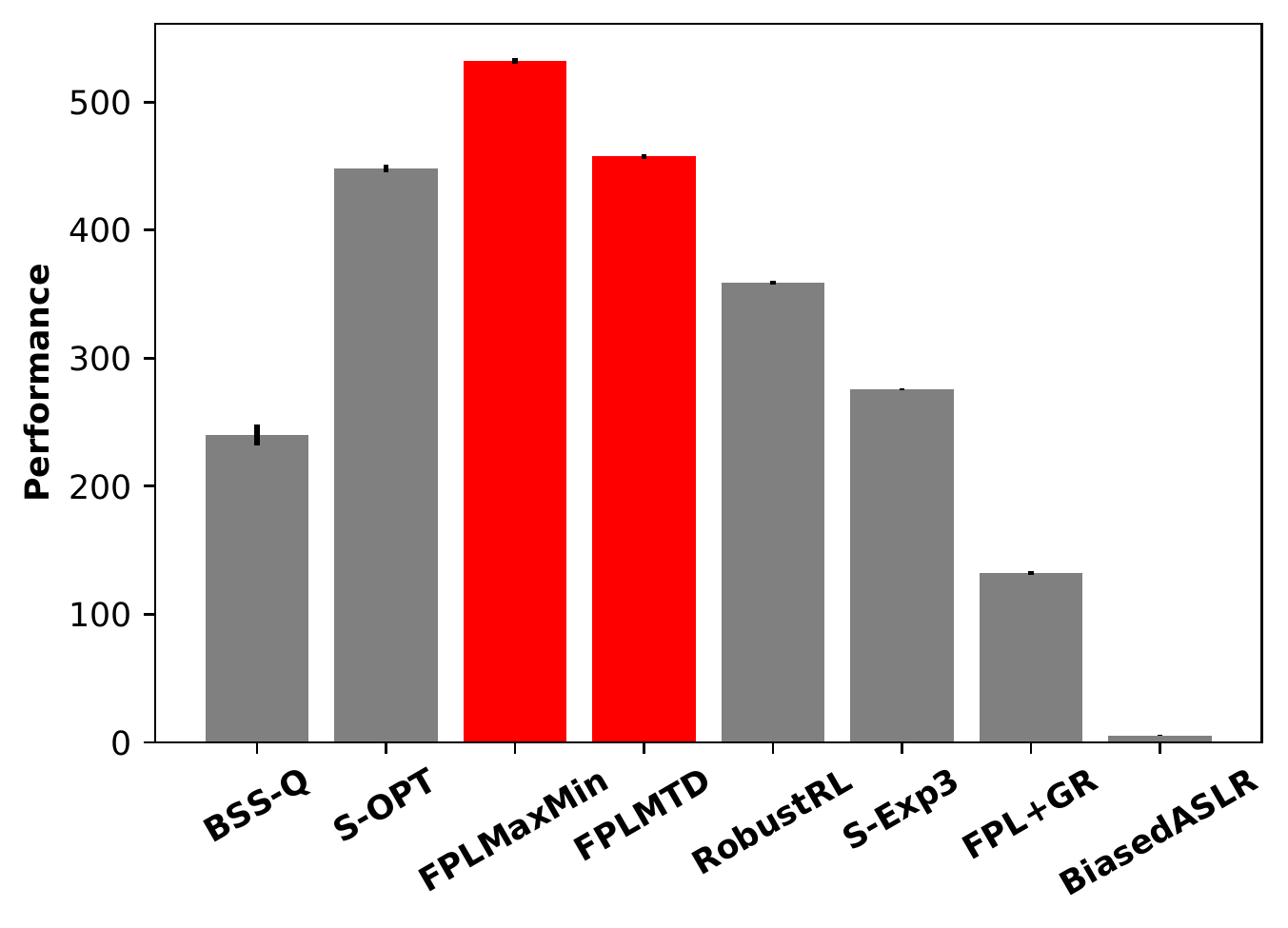}
         \caption{Random}
         \label{fig:general-sum-random}
     \end{subfigure}
        \caption{Performance of algorithms on general sum synthetic data. On each of the graphs, from left to right: BSS-Q, S-OPT, FPLMaxMin, FPLMTD, RobustRL, S-Exp3, FPL+GR}
        \label{fig:general-sum}
\end{figure*}

\begin{figure*}[!tp]
     \centering
     \begin{subfigure}[b]{0.24\textwidth}
         \centering
         \includegraphics[width=\textwidth]{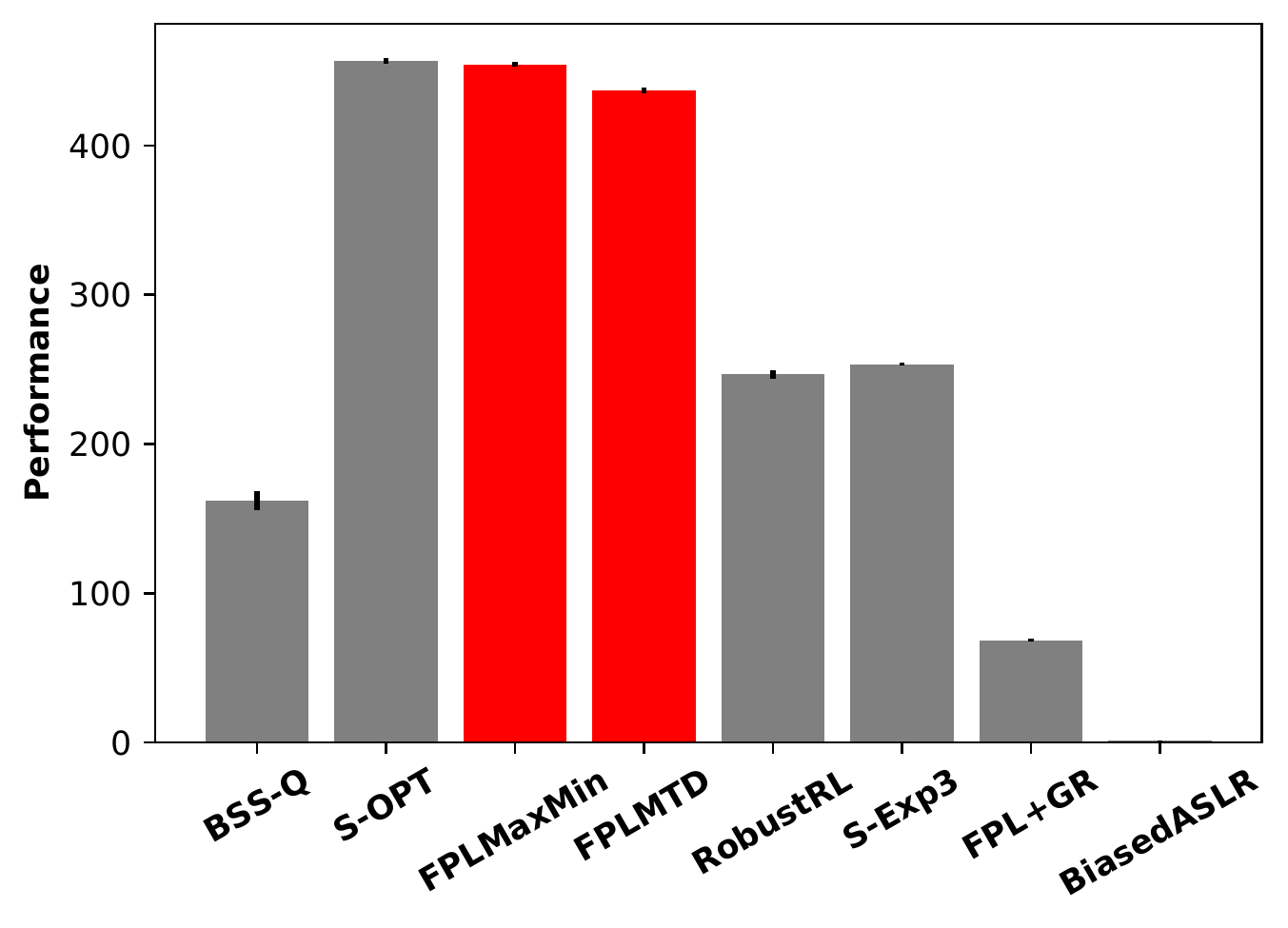}
         \caption{Best Response}
         \label{fig:zero-sum-best-response}
     \end{subfigure}
     \hfill
     \begin{subfigure}[b]{0.24\textwidth}
         \centering
         \includegraphics[width=\textwidth]{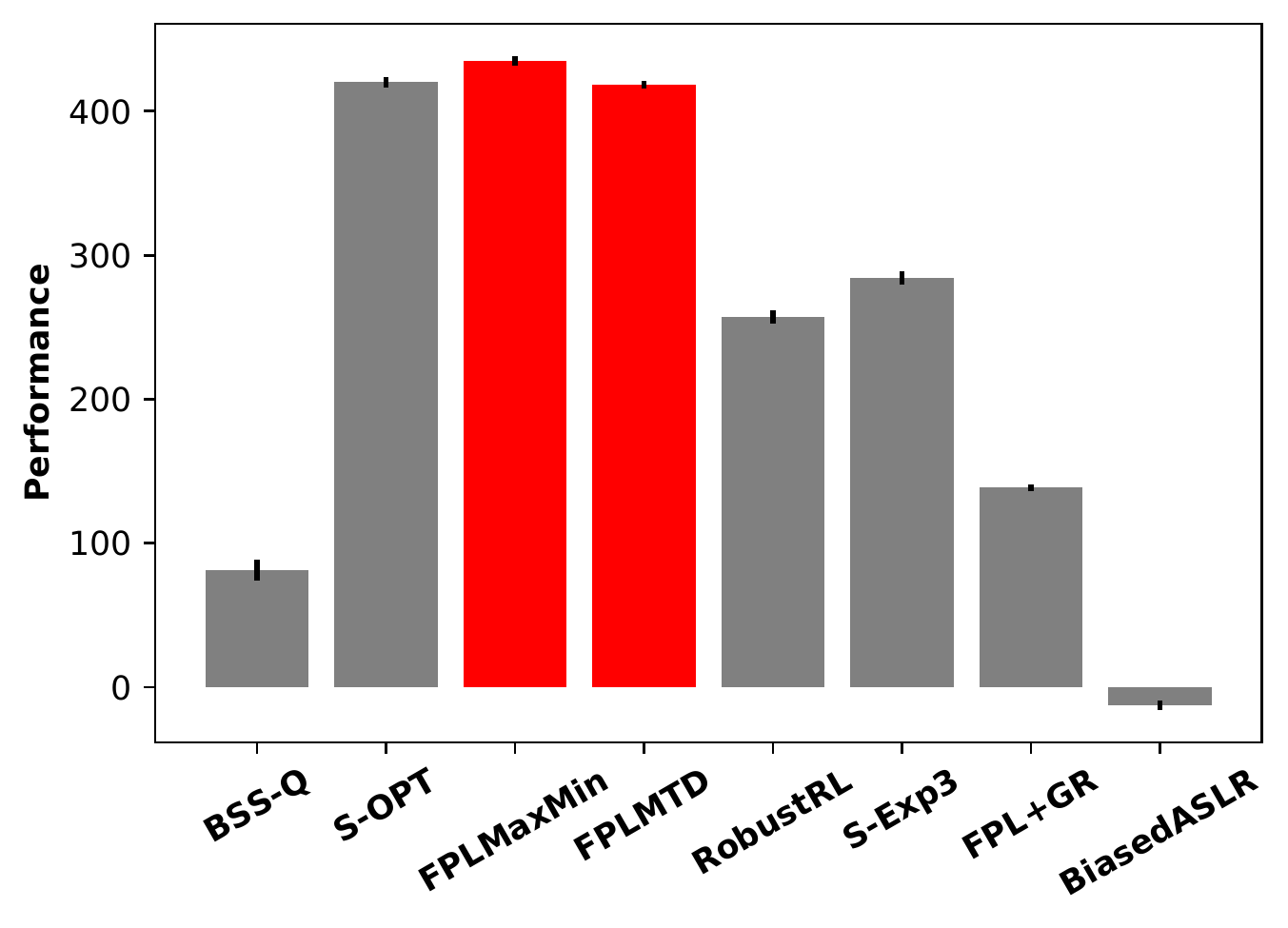}
         \caption{FPL-UE}
         \label{fig:zero-sum-fpl-ue}
     \end{subfigure}
     \hfill
     \begin{subfigure}[b]{0.24\textwidth}
         \centering
         \includegraphics[width=\textwidth]{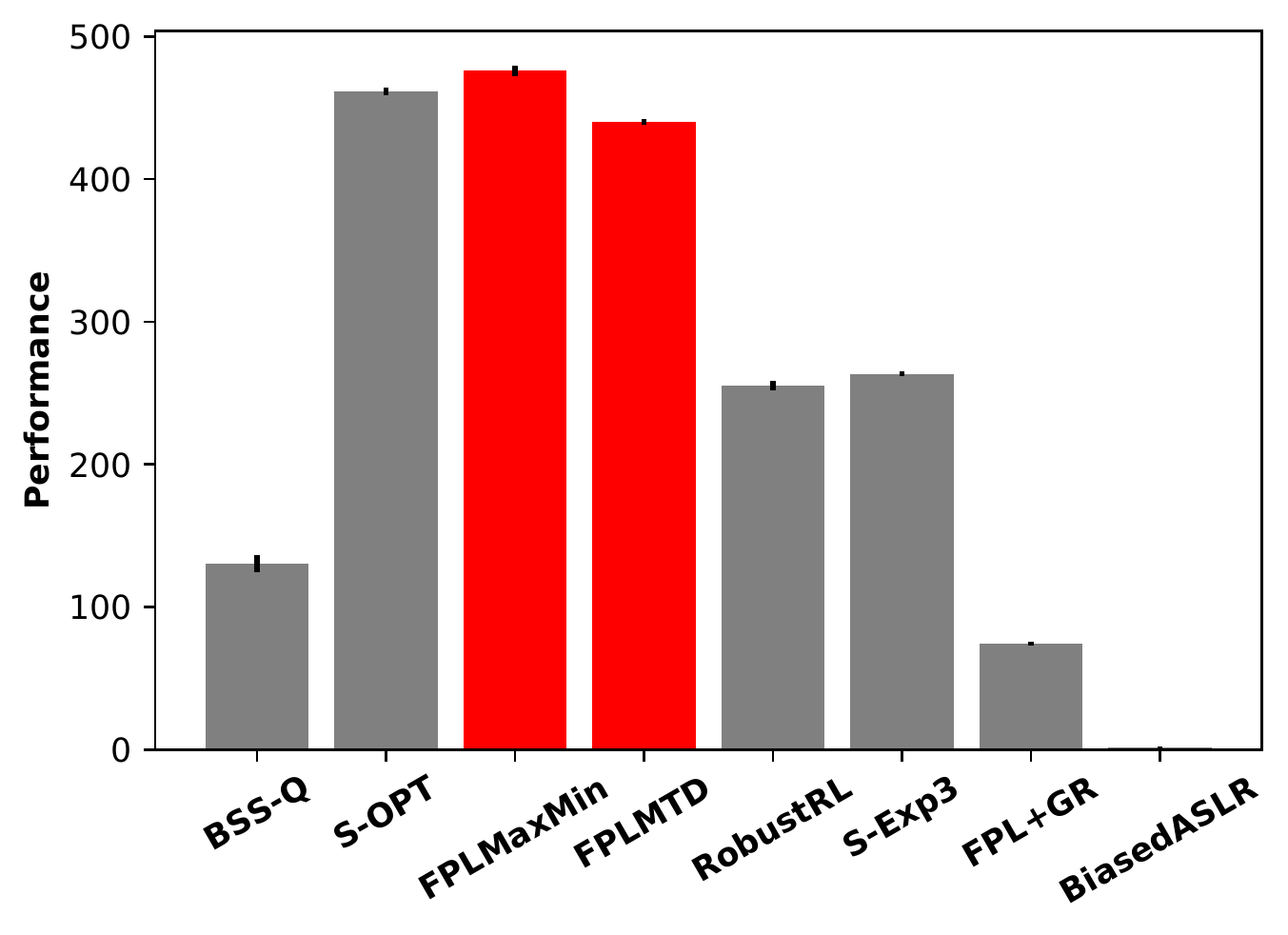}
         \caption{Stackelberg}
         \label{fig:zero-sum-stackelberg}
     \end{subfigure}
     \hfill
     \begin{subfigure}[b]{0.24\textwidth}
         \centering
         \includegraphics[width=\textwidth]{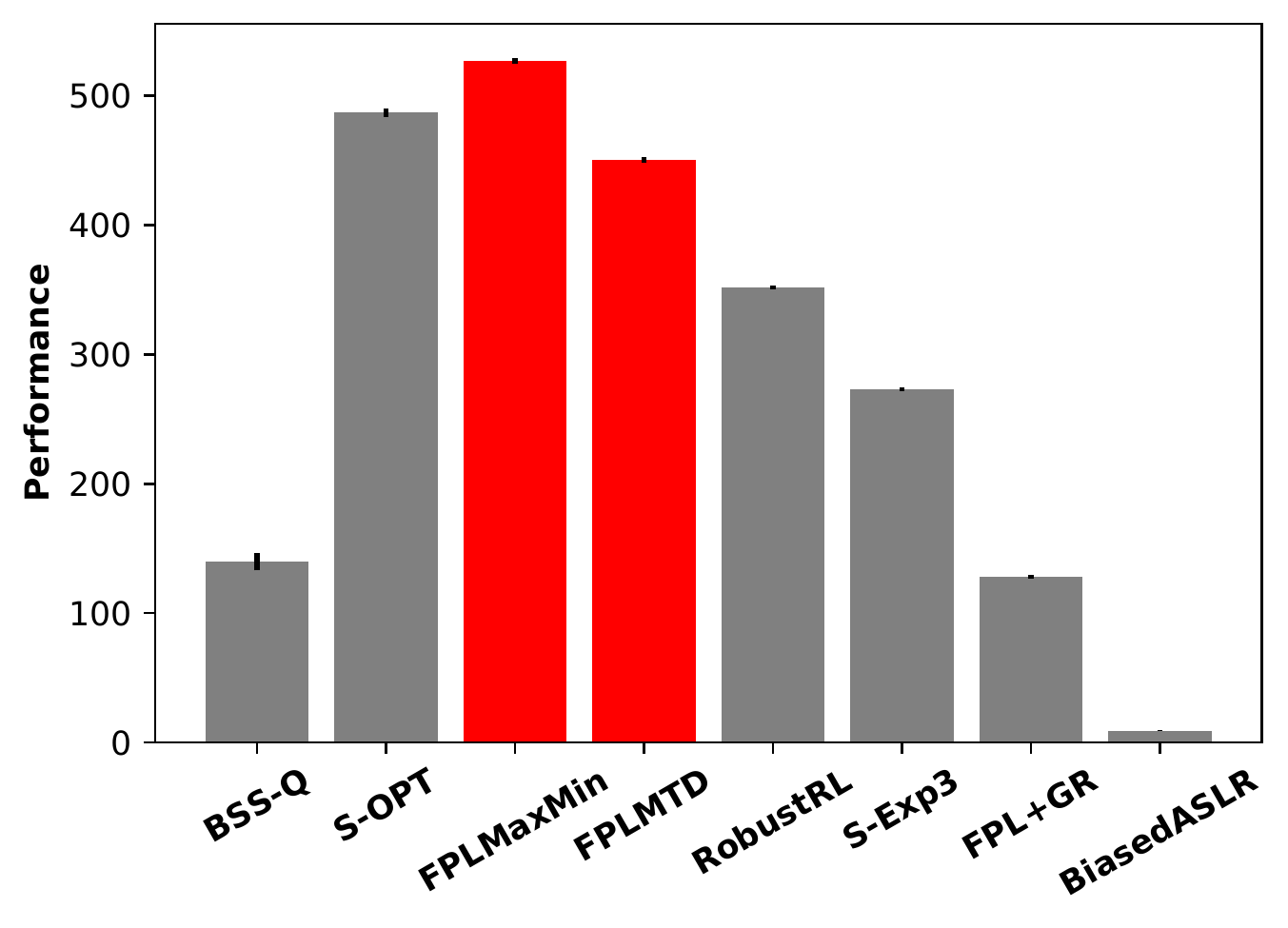}
         \caption{Random}
         \label{fig:zero-sum-random}
     \end{subfigure}
        \caption{Performance of algorithms on zero sum synthetic data. On each of the graphs, from left to right: BSS-Q, S-OPT, FPLMaxMin, FPLMTD, RobustRL, S-Exp3, FPL+GR}
        \label{fig:zero-sum}
\end{figure*}

\section{Identification of Critical Vulnerabilities}\label{sec:identification}

Along with providing a good switching strategy to facilitate Moving Target Defense, the reward estimates used by \emph{FPL-MaxMin} can also be used to help decide which vulnerabilities to fix when defenders do not have the resources to fix all the vulnerabilities of their configurations. 
More formally, we study the following problem in this section: given a budget $b$ and a price vector $\vec p$ where $p_v$ denotes the price of fixing vulnerability $v \in  V$, what is the best set of critical vulnerabilities that the defender should fix in order to improve the future performance of \emph{FPL-MaxMin}? Note that this cannot be done with \emph{FPL-MTD} since it does not maintain a reward estimate for each vulnerability. While the reward estimates of \emph{FPL-MTD} can be used to provide direction when deciding which configurations to focus on, it cannot be used to identify critical vulnerabilities. 

Let $r_{v, \psi}$ denote the reward estimate for each vulnerability-attacker type pair $(v, \psi)$ obtained by the \emph{FPL-MaxMin} algorithm. 
Given a budget $b$ and a price vector $\vec p$, one approach to choosing critical vulnerabilities is to follow the configuration selection approach of \emph{FPL-MaxMin} and choose vulnerabilities so as to maximize the weighted max-min value of the defender given by \eqref{eq:max-min}. More formally, let $S \subseteq V$ be the set of chosen vulnerabilities to be fixed.
Let $r^S$ denote the reward estimate where $S$ is fixed and the defender obtains no negative reward from that vulnerability anymore i.e. 
\begin{align*}
    r^S_{v, \psi} = 
    \begin{cases}
    r_{v, \psi} & v \notin S \\
    0 & v \in S
    \end{cases}
\end{align*}
The goal of the above mentioned approach is to choose a set $S$ such that it maximizes the weighted max-min value of the game subject to budget constraints. This can be represented by the following optimization problem:
\begin{align}
    \max_{S \subseteq  V} & \max_{c \in  C} \bigg (\sum_{\psi \in \Psi}  P_{\psi} \min_{v \in  V_c} r^S_{v, \psi} \bigg ) \notag \\
    \text{s.t.} & \sum_{v \in S} p_v \le b \label{eq:max-max-min}
\end{align}
where $ V_c$ is the set of vulnerabilities of configuration $c$.

It is easy to see that a set $S$ which maximizes the above optimization problem will be incredibly biased towards a  configuration which already has a high worst-case reward; this approach chooses a configuration with the highest worst-case reward and spends the entire budget fixing vulnerabilities of that configuration. This can be seen more clearly in Example \ref{ex:vul-choice}.

\begin{example}\label{ex:vul-choice}
Consider a setting with $1$ attacker type, $2$ configurations $\{c_1, c_2\}$ and $3$ vulnerabilities, $ V = \{v_1, v_2, v_3\}$. Let $c_1$ have vulnerabilities $\{v_1, v_2\}$ and $c_2$ have vulnerabilities $\{v_2, v_3\}$. Let the reward estimate vector be $r = (-1, -0.5, -0.1)$. From this, the worst vulnerability of $c_1$ has reward $-1$ and the worst vulnerability of $c_2$ has reward $-0.5$ Assume that there is enough budget to fix only one vulnerability which can be any of the three.

It is easy to see that \eqref{eq:max-max-min} will choose to fix $v_2$ and increase the reward of the worst vulnerability of $c_2$ to $-0.1$ while keeping the reward of the worst vulnerability of $c_1$ at $-1$.
\end{example}

Effective Moving Target Defense requires configurations which are roughly equal in terms of the number of vulnerabilities they have. Otherwise, when one configuration becomes significantly better than the others, the others will not be deployed as often and will go to waste. We, therefore, need to modify our approach to evenly develop each of the configurations. We can do this by replacing the $\max$ over the set of configurations with a $\min$. This intuitively maximizes the worst vulnerability of the worst configuration; such an approach spreads the budget across configurations since the ``worst'' configuration keeps changing as we fix vulnerabilities in some of them and make them better. 

More formally, we have a new optimization problem, which we call $\ChooseVul$:
\begin{align}
    \max_{S \subseteq  V} &  \min_{c \in  C} \bigg (\sum_{\psi \in \Psi}  P_{\psi} \min_{v \in  V_c} r^S_{v, \psi} \bigg ) \notag\\
    \text{s.t} & \sum_{v \in S} p_v \le b \label{eq:vulnerability-choice}
\end{align}
While the above formulation does model the problem at hand, it is unfortunately NP-Hard to solve even when there is only one configuration.

\begin{theorem}\label{thm:vulnerability-nphard}
Computing the optimal solution of $\ChooseVul$ is NP-Hard even when there is only one configuration.
\end{theorem}

\begin{proof}
We prove this via a reduction to the NP-Complete problem (0-1) KNAPSACK:
\begin{quote}
    Given $n$ items $N = [n]$, each with non-negative value $\val_i$ and non-negative weight $w_i$, and two non-negative numbers $W$ and $K$, does there exist a subset $S \subseteq N$ such that the total weight of $S$ is less than $W$ and the total value of $S$ is at least $K$?
\end{quote}
Given an instance of KNAPSACK, construct an instance of $\ChooseVul$ as follows: let there be one configuration with a set of $n$ vulnerabilities $\{v_1, v_2, \dots, v_n\}$ and a set of $n$ attacker types $\{\psi_1, \psi_2, \dots, \psi_n\}$, each with equal probability of attacking. Construct reward estimates as follows:
\begin{align*}
    r_{v_i, \psi_j} = 
    \begin{cases}
        -\val_i & i = j \\
        0 & i\ne j
    \end{cases}
\end{align*}
This can be intuitively understood as a setting where each attacker type can only exploit one vulnerability and no two types can exploit the same vulnerability.
Lastly, let the price vector $\vec p$ be equal to the weight vector $\vec w$ and the budget $b$ be equal to $W$.
With this instance formulation, $\ChooseVul$ given by \eqref{eq:vulnerability-choice} reduces to 
\begin{align}
    \max_{S \subseteq  V} &  \bigg (\sum_{i \in [n]} \frac1n  r^S_{v_i, \psi_i} \bigg ) \notag\\
    \text{s.t} & \sum_{v \in S} w_v \le W \label{eq:knapsack-choice}
\end{align}
This is because we can get rid of the $\min_{c \in  C}$ term by assuming only one configuration. We can also replace the $\min_{v \in  V_c}$ by the only vulnerability that the attacker type $\psi_i$ can exploit. 

Let us now denote the set $S$ by a binary vector $\vec x$ where $x_i = 1$ if and only if $i \in S$. We can re-write the above optimization problem as
\begin{align*}
    \max_{\vec x} &  \bigg (\sum_{i \in [n]} \frac{-\val_i + \val_i x_i}{n} \bigg ) \notag\\
    \text{s.t} & \sum_{i \in [n]} w_i x_i \le W \\
    & x_i \in \{0, 1\}
\end{align*}
By removing all the constant terms in the objective, the above problem is equivalent to the following problem
\begin{align}
    \max_{\vec x} &  \bigg (\sum_{i \in [n]}  \val_i x_i \bigg ) \notag\\
    \text{s.t} & \sum_{i \in [n]} w_i x_i \le W \notag \\
    & x_i \in \{0, 1\} \label{eq:knapsack-milp}
\end{align}
Note that this is equivalent to the optimization version of the knapsack problem. It is easy to see that the objective function value in \eqref{eq:knapsack-milp} is at least K if and only if the solution to the original knapsack instance is YES. Moreover, the solution to \eqref{eq:knapsack-milp} can be computed in polynomial time using the solution from \eqref{eq:knapsack-choice}. This completes the reduction.
\end{proof}

The above reduction requires multiple attacker types to exist. If we restrict the problem space to instances that only have one attacker type, the greedy algorithm can solve $\ChooseVul$ in polynomial time. The steps of the algorithm have been described in Algorithm \ref{algo:greedy-oneattacker}. The algorithm starts with an empty set and at each iteration, adds the vulnerability with the least reward estimate. It keeps going till it cannot add another item without violating the budget constraint.

\begin{algorithm}
    \caption{Greedy Algorithm for Vulnerability Selection with One Attacker Type}
    \label{algo:greedy-oneattacker}
    \begin{algorithmic}
        \State $S \gets \emptyset$ 
        \While{$1$}
            \State $v' =  \argmin_{v \in  V} r^S_{v, \psi}$ \Comment{If there are multiple, choose the one with the least price}
            \If{$\sum_{v \in S} p_v + p_{v'} \le b$}
                \State $S \gets S \cup \{v'\}$
            \Else 
                \State \textbf{break}
            \EndIf
        \EndWhile
        \State \textbf{return} $S$
    \end{algorithmic}
\end{algorithm}

\begin{theorem}\label{thm:greedy-oneattacker}
Algorithm \ref{algo:greedy-oneattacker} returns the optimal solution of $\ChooseVul$ when there is only one attacker type.
\end{theorem}
\begin{proof}
When there is only one attacker type (say $\psi$), $\ChooseVul$ reduces to the following optimization problem:
\begin{align}
    \max_{S \subseteq  V} &  \bigg ( \min_{v \in  V} r^S_{v, \psi} \bigg ) \notag\\
    \text{s.t} & \sum_{v \in S} p_v \le b \label{eq:oneattacker-objective}
\end{align}
This is because, with one attacker type, we can conflate the two $\min$ functions since there will not be any probability term between them. More specifically, we have $\min_{c \in  C} \min_{v \in  V_c} = \min_{v \in  V}$.

It is easy to see that for such a problem, the optimal solution picks the worst vulnerabilities while respecting the budget constraint. Assume for contradiction that this is not the case. Let $S_g$ be the suboptimal greedy solution and $S_{OPT}$ be the optimal solution. Let the final objective value of the greedy solution be $O_{g}$ and that of the optimal solution be $O_{OPT}$. Assume $O_g < O_{OPT}$.

Assume there exists some vulnerability $v \in S_g \setminus S_{OPT}$ with reward estimate $r_v$. Note that the greedy algorithm by definition picks all the vulnerabilities with a worse reward estimate and therefore, $O_g \ge r_v$. The optimal solution on the other hand, has an objective function upper bounded at $r_v$. This gives us $O_g \ge O_{OPT}$, a contradiction. 

Therefore, we must have $S_g \subseteq S_{OPT}$. Note that $S_g$ terminated since it could not add the least price minimum reward estimate vulnerability to the set due to the budget constraints. Let us call this vulnerability $v$ and denote its reward estimate by $r_v$. We have, by definition $O_g = r_v$. We also have $O_{OPT} = r_v$ since $OPT$ could not add $v$ but by being adding superset of $S_g$, it adds all the vulnerabilities worse than $v$. This gives us $O_g = O_{OPT}$, another contradiction.

Therefore, the greedy algorithm must return the optimal solution. 
\end{proof}

\begin{algorithm}
    \caption{Greedy Algorithm for Vulnerability Selection with Multiple Attacker Types}
    \label{algo:greedy}
    \begin{algorithmic}
        \State $S \gets \emptyset$ 
        \While{$1$}
            \State $S' =  \argmax_{S: |S|= 1}\min_{c \in  C} \big (\sum_{\psi \in \Psi} P_{\psi}\min_{v \in  V_c} r^{S\cup S'}_{v, \psi}$ \big ) \Comment{If there are multiple, choose the one with the least price}
            \If{$\sum_{v \in S( \cup S')} p_v \le b$}
                \State $S \gets S \cup S'$
            \Else 
                \State \textbf{break}
            \EndIf
        \EndWhile
        \State \textbf{return} $S$
    \end{algorithmic}
\end{algorithm}

For more complex cases which have multiple attacker types, we propose the use of a similar greedy algorithm (described in Algorithm \ref{algo:greedy}). At each iteration, the algorithm looks at all possible singleton sets of vulnerabilities and chooses the one which maximizes the objective function. It stops when the budget constraint is violated. This algorithm reduces to Algorithm \ref{algo:greedy-oneattacker} when there is only one attacker type.
While we do not have approximation guarantees for this algorithm, simple simulations show us that the greedy approach gives us a significant improvement in the total utility of \emph{FPL-MaxMin} after fixing only a small number of vulnerabilities. 

More specifically, we tested our approach on the Small NVD-Based Dataset (Section \ref{subsec:expts-nvd-small}) against the Best Response attacker. We assumed all vulnerabilities have a unit price and varied the total budget by trying all the budgets in the set $\{2, 5, 10, 15\}$. Note that with a budget of $15$, any algorithm can only choose roughly $5\%$ of the vulnerabilities.
We compared our approach against a random approach where vulnerabilities are chosen to be fixed randomly till the budget is exceeded.

In order to compare the two approaches, we first ran \emph{FPL-MaxMin} for $1000$ timesteps and then created three branches: one where the greedy algorithm (Algorithm \ref{algo:greedy}) was used to fix vulnerabilities, one where the random approach was used to fix vulnerabilities and one where no vulnerabilities were fixed. For all the three branches, we continued to use the same \emph{FPL-MaxMin} instance for $1000$ more timesteps. We performed this experiment $50$ times and recorded the total utilities.

\begin{figure*}[!tp]
     \centering
     \begin{subfigure}[b]{0.24\textwidth}
         \centering
         \includegraphics[width=\textwidth]{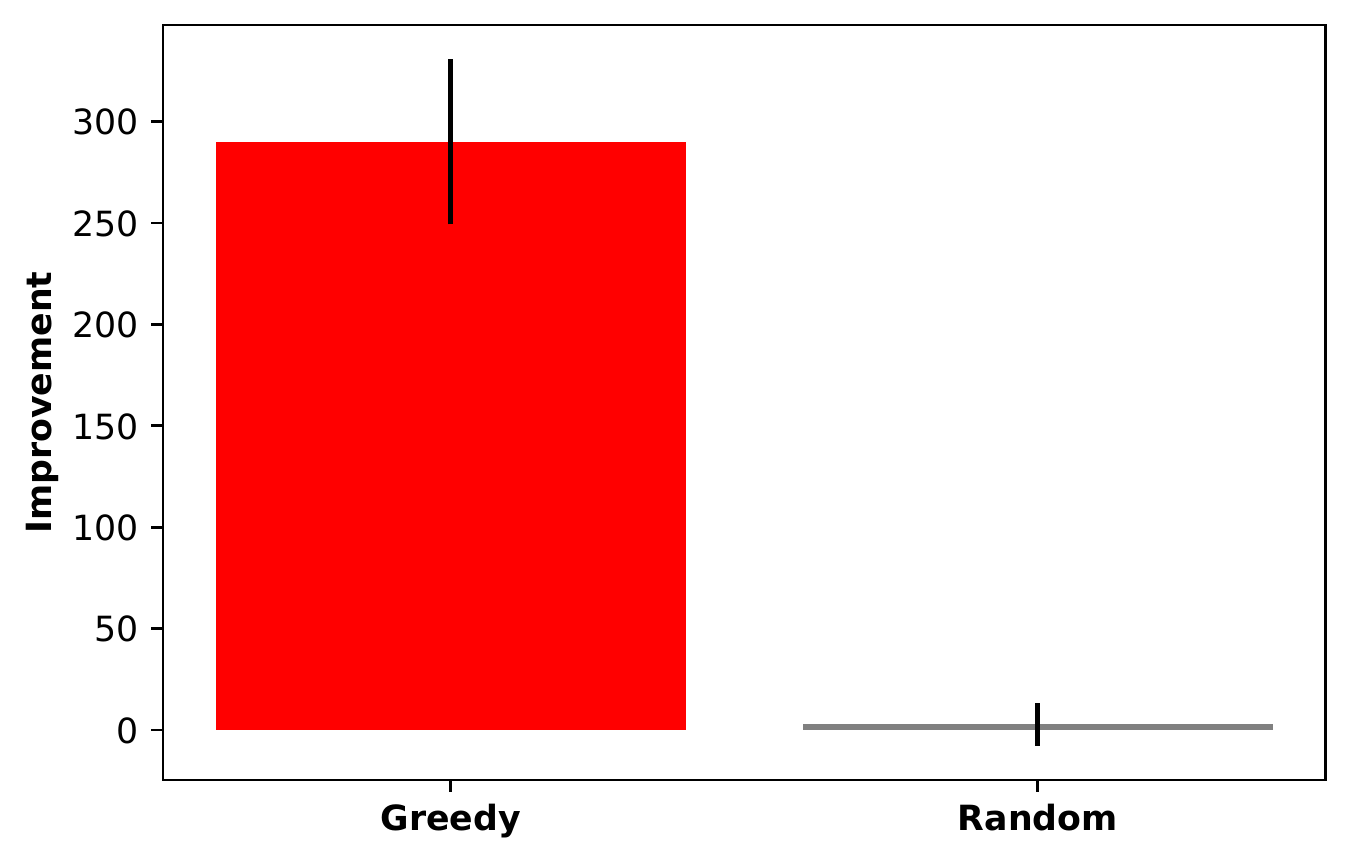}
         \caption{Budget $= 2$}
         \label{fig:critical-vul-b2}
     \end{subfigure}
     \hfill
     \begin{subfigure}[b]{0.24\textwidth}
         \centering
         \includegraphics[width=\textwidth]{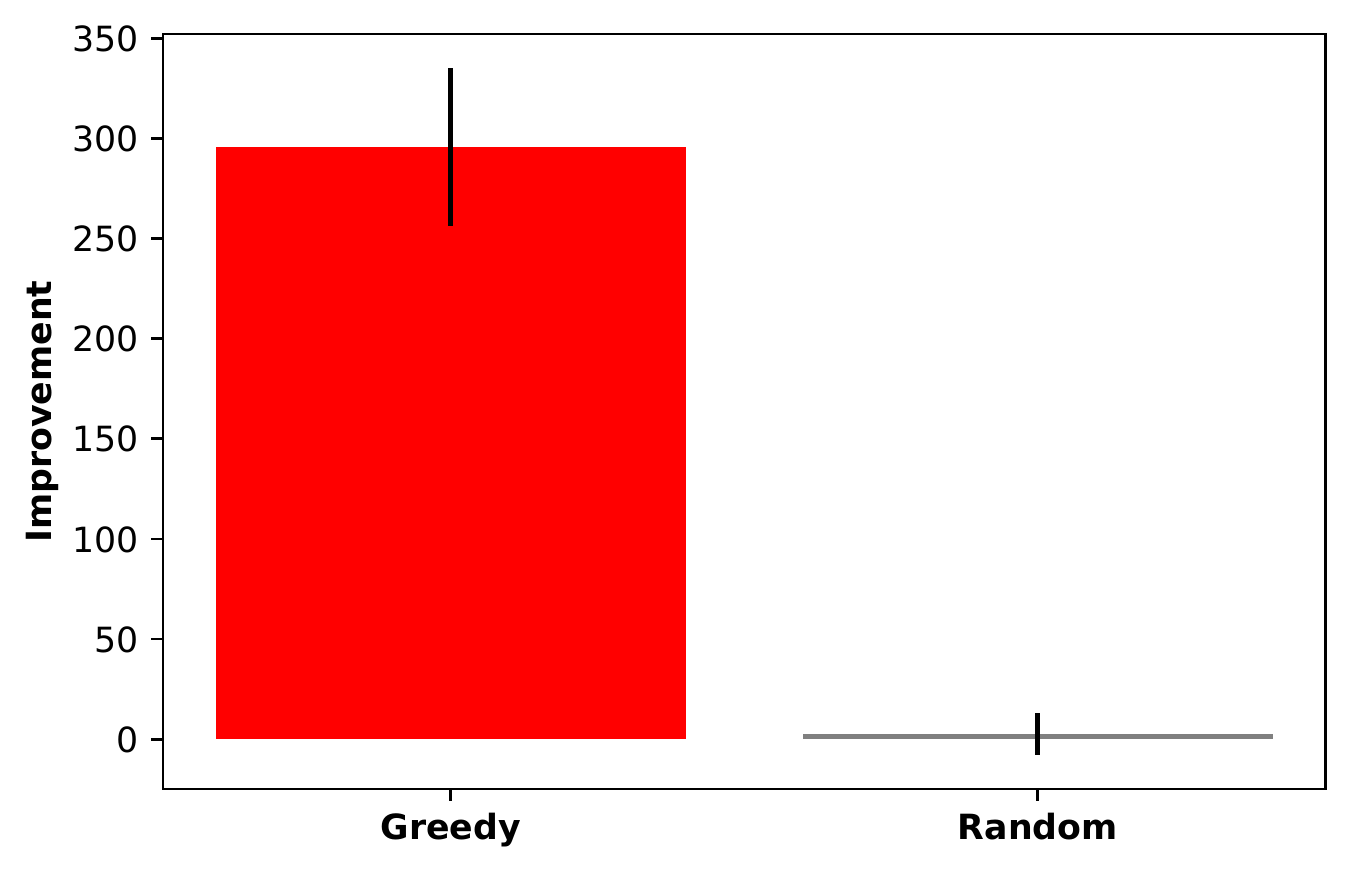}
         \caption{Budget $= 5$}
         \label{fig:critical-vul-b5}
     \end{subfigure}
     \hfill
     \begin{subfigure}[b]{0.24\textwidth}
         \centering
         \includegraphics[width=\textwidth]{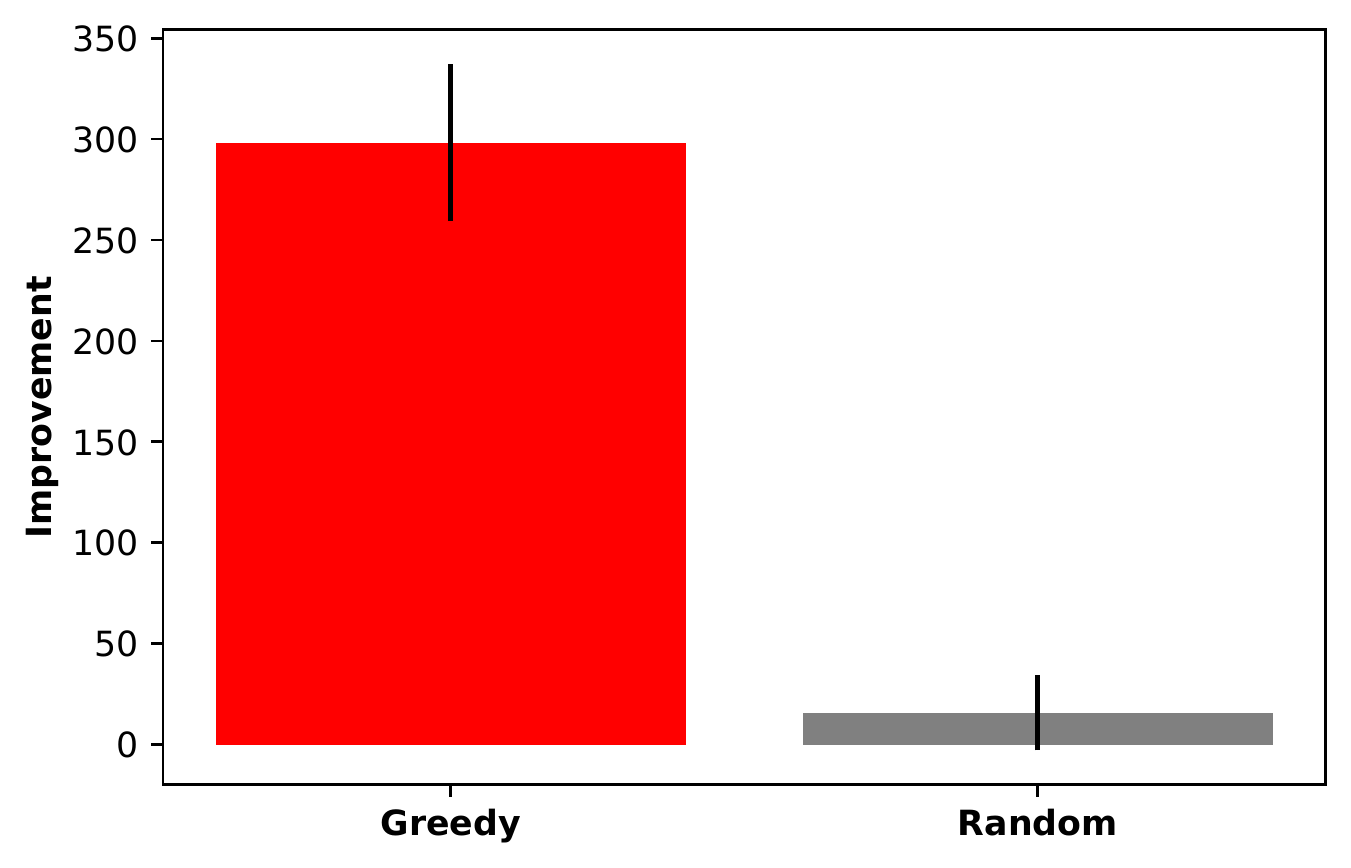}
         \caption{Budget $= 10$}
         \label{fig:critical-vul-b10}
     \end{subfigure}
     \hfill
     \begin{subfigure}[b]{0.24\textwidth}
         \centering
         \includegraphics[width=\textwidth]{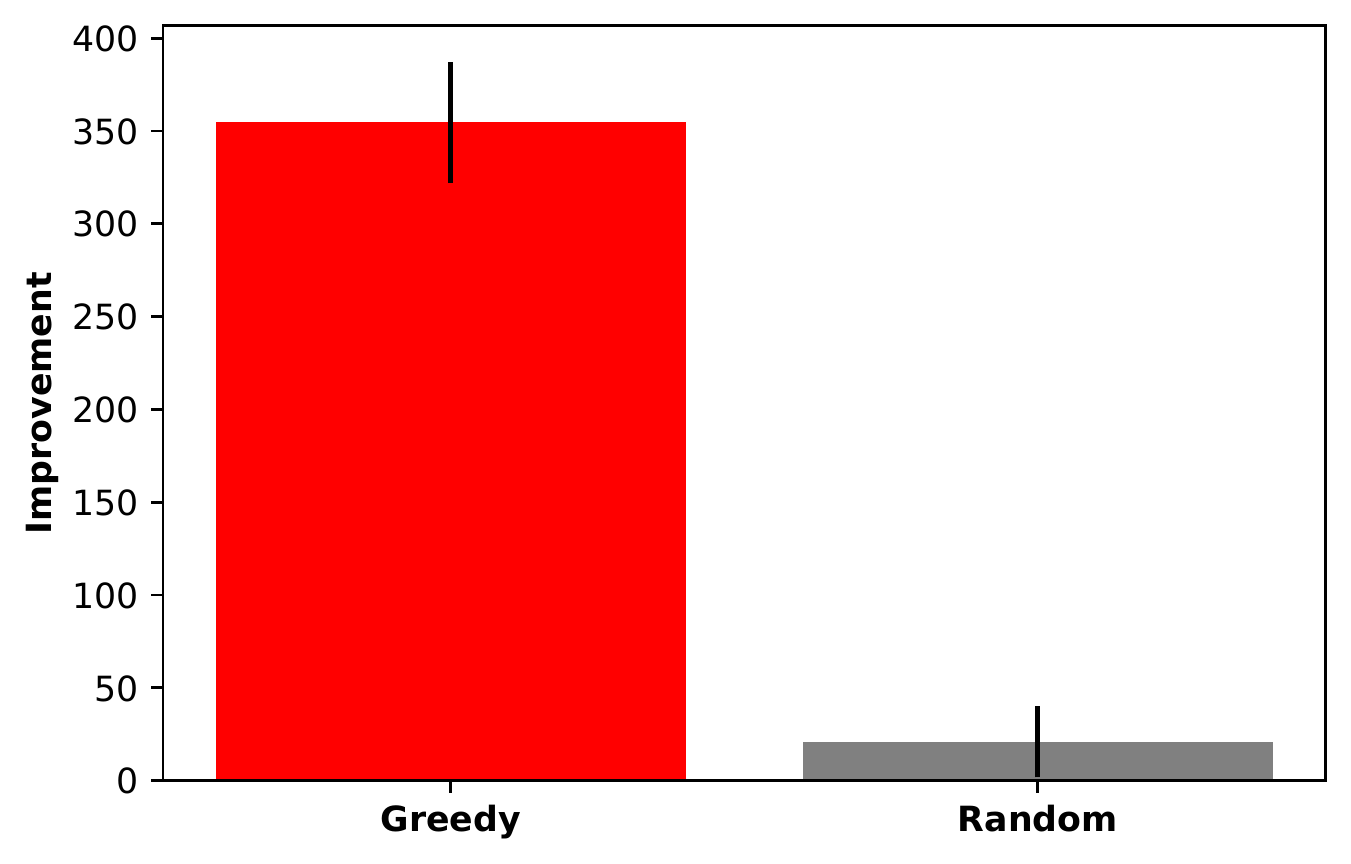}
         \caption{Budget $= 15$}
         \label{fig:critical-vul-b15}
     \end{subfigure}
        \caption{Improvement of vulnerability selection algorithms on the small NVD-based dataset. The greedy approach is on the left and the random approach is on the right.}
        \label{fig:critical-vul}
\end{figure*}

We define {\em improvement} as the difference between the total utility of the algorithm after fixing vulnerabilities and the total utility of the algorithm when no vulnerabilities were fixed. We plotted the average improvement along with standard error bounds for each of our vulnerability fixing approaches in Figure \ref{fig:critical-vul}. 

Our results show that the greedy approach performs significantly
better than the random approach, showing an improvement of more than $100$ times that of the random approach when the budget is equal to $5$. We also find that, when the budget is $5$, the total utility after fixing the vulnerabilities reduces to an average of $57.1\%$ of the total utility without fixing any vulnerabilities; this is good since the total utility is always negative. This shows that we can obtain significant improvements by intelligently choosing the vulnerabilities to fix.

While $\ChooseVul$ is the problem formulation we use, it is just one of the many objectives that can be maximized when choosing vulnerabilities. It arguably may not be optimal in all the settings it is applied to but for this problem, the definition of optimal is itself contentious. We leave a more detailed study of this problem and its various approaches for future work.
\section{Technical Details of Experiments}\label{sec:technical-details}
\subsection{Dataset Generation}\label{subsec:dataset-generation}
\subsubsection{Generation of Large NVD-Based Data}\label{subsub:nvd-large-data}
The National Vulnerability Database (NVD) \footnote{nvd.nist.gov} is a public directory which contains information about Common Vulnerabilities and Exploits (CVEs) that might affect system components. This directory is updated with a new file every year.
For each vulnerability, they use a scoring system called the Common Vulnerability Scoring System (CVSSv3) which score different aspects of the vulnerability. Out of the many scores they use, two scores which are important are the Impact Score ($\IS$) which scores the impact of a successful exploit and the Base Score ($\BS$) which combines the impact score with the ease of exploitability of the vulnerability.
\cite{Sailik2016Webappmtd} use these scores to generate rewards in a Bayesian game as follows: for each vulnerability $v \in  V$, for each attacker type $\psi \in \Psi$ and for each configuration $c \in  C$ they generate a reward function
\begin{align}
    r_{\theta}(\psi, v, c) = 
    \begin{cases}
        -IS_v & v \in  V_{c} \text{ and the attacker type $\psi$ can exploit $v$} \\
        0 & \text{ Otherwise}
    \end{cases}
    \notag \\
    r_{\psi}(v, c) = 
    \begin{cases}
        BS_v & v \in  V_{c} \text{ and the attacker type $\psi$ can exploit $v$} \\
        0 & \text{ Otherwise}
    \end{cases} \label{eq:nvd-large-reward}
\end{align}
where $IS_v$ and $BS_v$ are the impact score and base score of $v$ respectively.

We use the same approach to convert vulnerability scores to rewards; the only difference being that we scale the rewards down rewards to ensure they are in the range $[0, 1]$.
We mine the Common Vulnerabilities and Exploits (CVE) data for the years $2002$ to $2021$ to collect a database of $87180$ vulnerabilities along with their impact scores and base scores. We call this set of vulnerabilities the {\em global set of vulnerabilities}. 

For each dataset, we choose the number of configurations uniformly at random between $10$ and $20$, the number of attackers uniformly at random between $3$ and $6$ and the number of vulnerabilities uniformly at random between $500$ and $800$. We sample each vulnerability in the dataset from the global set of vulnerabilities uniformly at random. 

For each attacker type, we sample a skill level from a truncated normal distribution with range $[0, 1]$. We then obtain the number of vulnerabilities that the attacker can exploit using a truncated normal distribution with mean proportional to their skill level and we choose these vulnerabilities uniformly at random. For each configuration, we generate a vulnerability set such that each vulnerability is not a vulnerability of the configuration with probability $0.01$. 

Using the base score and the impact score of each vulnerability along with the vulnerability set for each configuration and the set of vulnerabilities each attacker can exploit, we generate the reward function for the defender and the attacker according to \eqref{eq:nvd-large-reward}.


\subsubsection{Generation of Synthetic General and Zero Sum Data}\label{subsub:synthetic-data}
For each dataset, we choose the number of configurations uniformly at random between $10$ and $20$, the number of attackers uniformly at random between $3$ and $6$ and the number of vulnerabilities uniformly at random between $500$ and $800$.  

For each attacker type, we sample a skill level from a truncated normal distribution with range $[0, 1]$. We then obtain the number of vulnerabilities that the attacker can exploit using a truncated normal distribution with mean proportional to their skill level and we choose these vulnerabilities uniformly at random. For each configuration, we generate a vulnerability set such that each vulnerability is not a vulnerability of the configuration with probability $0.05$. 

For each vulnerability, we sample the defender reward from $U[-1,0]$ and the attacker reward from $U[0,1]$. In summary, we create the reward function as follows:
\begin{align*}
    r_{\theta}(\psi, v, c) = 
    \begin{cases}
        U[-1, 0] & v \in  V_{c} \text{ and the attacker type $\psi$ can exploit $v$} \\
        0 & \text{ Otherwise}
    \end{cases}
    \\
    r_{\psi}(v, c) = 
    \begin{cases}
        U[0, 1] & v \in  V_{c} \text{ and the attacker type $\psi$ can exploit $v$} \\
        0 & \text{ Otherwise}
    \end{cases}
\end{align*}
For zero sum datasets, we only sample the defender reward and take the attacker reward as the additive inverse of the defender reward.

\subsection{Computational Resources}\label{subsec:resources}
All the Mixed Integer Quadratic Programs were solved using Gurobi. All the experiments were run on an Intel Xeon E5-2680 v4 processor with 128GB RAM.


\subsection{Reproducibility}\label{subsec:reproducibility}
In order to ensure the reproducibility of the code, a seed value of $2022$ has been used to initialize a pseudo-random number generator in all our programs. The same generator is passed to all the functions that use randomization. All the datasets used as well as the values of the hyperparameters have been included in the code; these values have also been presented in Table \ref{table:parameters}.

\begin{table}
\centering
\begin{tabular}{|c | c|} 
 \hline
 Parameter & Value \\
 \hline
    Seed & 2022\\
    Number of attacker types & [3, 6]\\
    Number of vulnerabilities & [500, 800]\\
    Number of configurations & [10, 20]\\
    T & 1000\\
    $\gamma_{MaxMin}$ & 0.006\\
    $\eta_{MaxMin}$ & 0.03\\
    $\gamma_{MTD}$ & 0.007\\
    $\eta_{MTD}$ & 0.1\\
    $\alpha_{BSSQ}$ & 0.2\\
    Discount factor ($\gamma_{BSSQ}$) & 0.8\\
    $\epsilon_{RobustRL}$ & 0.1 \\ [1ex] 
 \hline
\end{tabular}
\caption{Parameter Values}
\label{table:parameters}
\end{table}

\end{document}